%% file: Source_Files.tex
\documentclass[11pt]{article}
\input{qudit_preamble.tex}

\begin{document}

\title{\bf Revisiting the invariant ring of two-qubit mixed states}

\author{Bing Xie\orcidlink{0009-0005-7767-2993}\footnote{E-mail: xiebingjiangxi2023@163.com}\quad and\quad Lin Zhang\orcidlink{0000-0001-6220-4218}\footnote{Contact author: godyalin@163.com}\\
{\it\small School of Mathematical Sciences, Hangzhou Dianzi University, Hangzhou 310018, China}}
\date{}
\maketitle

\begin{abstract}

Local unitary (LU) equivalence classifies quantum states up to changes of local bases and preserves their entanglement properties. Polynomial LU invariants provide an algebraic approach to determining whether two quantum states are LU equivalent. We revisit the ring of complex polynomial LU invariants on the full two-qubit operator space, equivalently the invariant ring under the local conjugation action of $\GL(2,\bbC)\times\GL(2,\bbC)$. For the two-qubit system, King \emph{et al.} previously determined a generating set of 21 invariants and the Cohen--Macaulay structure of this ring. Using a unified notation, we provide a self-contained derivation of the known Molien series and a systematic graphical reconstruction of the known generators. The Molien series is derived by explicit contour integration over the maximal torus, supplying computational steps not given in the earlier treatment. We systematically reconstruct the known generators using the graphical tensor method and provide explicit algebraic reductions of the candidate invariants. This work clarifies the calculations and graphical construction underlying the established results.
\end{abstract}

\newpage

\section{Introduction}

One of the most striking features of quantum mechanics is
entanglement in composite systems. Any local unitary (LU)
transformation merely changes the local bases of the subsystems.
It can neither eliminate nor create entanglement, nor alter its
degree. As a central resource in quantum information processing,
entanglement reflects the internal structure of composite quantum
systems. Consequently, any valid measure of entanglement must be
invariant under LU transformations \cite{Schlienz1995}. In general,
functions that remain unchanged under LU operations are called
LU invariants \cite{Zhang2026}. In this paper, we focus specifically
on polynomial invariants in the matrix elements, or equivalently
the Bloch parameters, of two-qubit operators. Following
King \emph{et al.}, we study these invariants under the local
conjugation action of $\GL_A(2)\times\GL_B(2)$, which yields the
same polynomial invariant ring as the local unitary action.
The set of all such polynomial invariants forms a ring, known as
the invariant ring. Constructing a complete generating set and
determining all algebraic relations among its generators amounts
to fully characterizing this ring \cite{Turner2017}.

It is noteworthy that for the two-qubit system, the complete
construction of LU invariants and the problem of state equivalence
admit fully explicit solutions. As the simplest nontrivial composite
quantum system, the two-qubit model has been extensively
investigated from multiple perspectives. Early work by Grassl
\emph{et al.} \cite{Grassl1998} established a connection between
polynomial invariants and representations of the permutation group,
reducing the construction of invariants to the enumeration of
trivial subrepresentations in tensor product representations. Within
this framework, they constructed a candidate generating set
consisting of 21 invariants and proved that this set generates all
invariants up to degree 23. At that stage, however, it had not yet been established whether
these 21 polynomials could distinguish all LU-inequivalent
two-qubit states or generate the entire invariant ring.

Subsequently, Makhlin~\cite{Makhlin2002} proved that a subset of 18
polynomial invariants is complete, in the sense that two two-qubit
mixed states are LU equivalent if and only if these 18 invariants take the same values for the two states. This subset distinguishes LU equivalence classes but does not generate the entire polynomial invariant ring. Meanwhile, King \emph{et al.} \cite{King2007} systematically studied this problem using graphical methods, obtaining the 21 generators of the invariant ring. They further identified 63 syzygies among these generators and
used them to determine the Cohen--Macaulay structure of the
invariant ring for the two-qubit system. This confirmed that the 21 invariants conjectured by Grassl \emph{et al.} suffice to generate the entire ring. Parallel frameworks have also been established, including Lie-algebraic PDE constructions \cite{Albeverio2007} and analyses of constrained states or special subclasses \cite{Gerdt2010,Gerdt2016,Candelori2024Symmetric}. From an experimental perspective, Zhang \emph{et al.} \cite{Zhang2025} established a correspondence between the Makhlin invariants and LU Bargmann invariants, enabling the measurement of these invariants through experimental schemes such as cycle tests \cite{Quek2024}. Randomised measurements provide another route to accessing two-qubit LU invariants, as reviewed in \cite{Cieslinski2024}.

Although the works mentioned above address the construction of LU
invariants from various perspectives, the work of King \emph{et al.}
is the most complete in terms of algebraic structure. They not only
provided an explicit construction of the generators, but also fully
characterized all the relations among them, thereby determining the
overall structure of the invariant ring. However, the original
derivation involves substantial details from representation theory
and computational algebra. The computation of the Molien series and the logic behind the graphical construction are not easily accessible to non-specialist readers. This paper aims to provide a systematic exposition of the work of King \emph{et al.}, focusing on the computation of the Molien series and the graphical tensor method. We supply the intermediate contour-integration steps and systematically reconstruct the known 21 generators. We also clarify a number of notational and terminological issues. It is our hope that this exposition makes this important
result more accessible to researchers in quantum information and
invariant theory.

The structure of this paper is as follows. Section~\ref{s2} reviews
the necessary preliminaries, including the Bloch expansion, the
vectorization framework, the reduction of the group action, the
basic concepts of invariant rings, and the definition of the Molien
series. Section~\ref{s3} presents a detailed computation of the
Molien series, providing a complete derivation from Molien's theorem
to the closed-form expression. Section~\ref{s4} introduces the
graphical tensor method, systematically explaining how to construct
invariants by enumerating connected graphs, and presents the
complete list of the 21 generators. Finally, Section~\ref{s5}
concludes the paper with a summary and a discussion of future
directions.

\section{Preliminaries}\label{s2}

We focus on the two-qubit system whose Hilbert space is given by $\cH_{AB}=\bbC^2\otimes\bbC^2$. Any operator $	\rho	=	\sum_{m,n,\mu,\nu} \innerm{m\mu}{\rho}{n\nu}\out{m\mu}{n\nu}\in\End(\cH_{AB})$ can be expanded in the Pauli basis as
$$\rho
=
t\I\otimes\I
+\bsa\cdot\pauli\otimes\I
+\I\otimes\bsb\cdot\pauli
+\sum_{\alpha,i=1}^{3}
c_{\alpha i}\sigma_\alpha\otimes\sigma_i,$$
where \(t\in\bbC\), \(\bsa=(a_1,a_2,a_3)^{\t},\bsb=(b_1,b_2,b_3)^{\t}\in\bbC^3\),
and	\(\bsC:=(c_{\alpha i})_{3\times3}\in\bbC^{3\times3}\).
Here $
\bsa\cdot\boldsymbol{\sigma}
=
\sum_{\alpha=1}^{3}a_\alpha\sigma_\alpha$,
and \(\sigma_\alpha\) are the Pauli matrices. The coefficients are given by $
t=\frac14\Tr{\rho},~
a_\alpha=\frac14\Tr{\rho(\sigma_\alpha\otimes\I)},~
b_i=\frac14\Tr{\rho(\I\otimes\sigma_i)},~
c_{\alpha i}
=
\frac14\Tr{\rho(\sigma_\alpha\otimes\sigma_i)}$.
If \(\rho\) is Hermitian, then \(t\in\bbR\), \(\bsa,\bsb\in\bbR^3\), and	\(\bsC\in\bbR^{3\times3}\). If, in addition, \(\rho\) is a density operator satisfying \(\Tr{\rho}=1\), then \(t=\frac14\), yielding the standard Bloch form.

Consider the action of \(\G:=\GL_A(2,\bbC)\times\GL_B(2,\bbC)\) on the full complex operator space \(\End(\bbC^2\ot\bbC^2)\), which is given by
\[
(g,h)\cdot\rho
=
(g\otimes h)\rho(g\otimes h)^{-1},
\qquad
g\in\GL_A(2),~
h\in\GL_B(2),~ \rho\in\End(\bbC^2\ot\bbC^2)
\]
where \(\GL_A(2)\) and \(\GL_B(2)\) denote two
copies of \(\GL(2,\bbC)\). In terms of
matrix elements, this transformation reads
$$
\innerm{m\mu}{\rho}{n\nu}\mapsto \innerm{m\mu}{(g \otimes h)\rho(g
\otimes h)^{-1}}{n\nu},
$$
where $\innerm{m\mu}{(g \otimes h)\rho(g \otimes h)^{-1}}{n\nu} =
\sum_{m', \mu', n', \nu'} g_{mm'} h_{\mu\mu'}
\innerm{m'\mu'}{\rho}{n'\nu'} (g^{-1})_{n'n} (h^{-1})_{\nu'\nu}$.

We adopt the following definition of vectorization for a bipartite
matrix:
$$
\Vec(\out{m}{n}\otimes \out{\mu}{\nu}) \defeq \Vec(\out{m}{n})
\otimes \Vec(\out{\mu}{\nu}) = \ket{mn} \otimes \ket{\mu\nu}.
$$
Consider the 16-dimensional vector $\Vec{(\rho)} =
(\innerm{m\mu}{\rho}{n\nu})_{mn\mu\nu}$ which is transformed under
the group $\G$ into another 16-dimensional
vector $\Vec(\rho') = (\innerm{m\mu}{\rho'}{n\nu})_{mn\mu\nu}$. By
direct computation,
\begin{eqnarray*}
\Vec(g\out{m}{n} g^{-1} \ot h\out{\mu}{\nu} h^{-1})
&=& \Vec(g\out{m}{n} g^{-1}) \ot \Vec(h\out{\mu}{\nu} h^{-1})\\
&=& (g \ot (g^{-1})^\t) \Vec(\out{m}{n}) \ot (h \ot (h^{-1})^\t) \Vec(\out{\mu}{\nu})\\
&=& (g \ot (g^{-1})^\t) \otimes (h \ot (h^{-1})^\t) \Vec(\out{m}{n}
\ot\out{\mu}{\nu}),
\end{eqnarray*}
which leads to
\begin{eqnarray}\label{zhi1}
\Vec(\rho') &=& \Vec{( (g \ot h)\rho(g \ot h)^{-1})}\notag\\
&=&\Br{(g \ot (g^{-1})^\t) \ot (h \ot (h^{-1})^\t)} \Vec(\rho).
\end{eqnarray}

Alternatively, we may vectorize the same operator by grouping the row and column indices separately. Define
\[
\Vec_{\mathrm r}\bigl(
|m\rangle\langle n|\otimes|\mu\rangle\langle\nu|
\bigr)
=|m\mu\rangle\otimes|n\nu\rangle.
\]
For this convention, the identity $\vec_{\mathrm r}{(\bsX\bsM\bsY)}=(\bsX\otimes\bsY^\t)\Vec_{\mathrm r}(\bsM)$
holds for $\bsX,\bsY,\bsM\in \bbC^{4\times 4}$. Consequently,
\begin{align}\label{zhi2}
\Vec_{\mathrm r}(\rho')
&=\Vec_{\mathrm r}\bigl(
(g\otimes h)\rho(g\otimes h)^{-1}
\bigr)\nonumber\\
&=\bigl[(g\otimes h)\otimes
((g^{-1})^\T\otimes(h^{-1})^\T)\bigr]
\Vec_{\mathrm r}(\rho).
\end{align}

Under the two vectorization conventions above, the same local conjugation action is represented in two different tensor-product forms. The corresponding orderings of the tensor factors match the two group--subgroup chains used in \cite{King2007}. Both branching rules ultimately restrict to the same {} group $\GL_A(2)
\times \GL_B(2)$, but differ in the intermediate steps regarding how
the indices are paired and whether the inverse representations are
explicitly written out. The two algebraic forms arising from the
vectorization framework are precisely the linear incarnations of
these two group-subgroup chain decompositions.

\begin{definition}
A \emph{representation} of a group $\G$ on a finite-dimensional
complex vector space $\cX$ is a homomorphism $\pi : \G \to
\GL(\cX)$ from $\G$ to the general linear group of $\cX$, i.e. the
invertible elements in $\End(\cX)$, all complex linear maps from
$\cX$ to itself.
\end{definition}

Let $V=\Vec\bigl(\End(\bbC^2\ot\bbC^2)\bigr)$ be the vector space
of vectorized operators acting on $\bbC^2\ot\bbC^2$.
Since vectorization is a linear isomorphism, we have
$V\cong\bbC^{16}$. The action of the group $\G := \GL_A(2) \times \GL_B(2)
$ on operators induces a representation
$$
\bsT : \G \to \GL(V) \cong \GL(16),
$$
which maps each group element $(g, h)$ to an invertible linear
transformation on $V$. From (\ref{zhi1}), it holds that $\bsT : (g,
h) \mapsto (g \ot (g^{-1})^\t) \ot (h \ot (h^{-1})^\t)$. Let
$\bsT_{(g,h)}:=\bsT(g,h)$. Then $\bsT_{(g,h)} = (g \ot (g^{-1})^\t)
\ot (h \ot (h^{-1})^\t)$.

Recall that $\set{x_1, \dots, x_{16}}$ is a basis of the dual space
$V^*$, defined by $x_i(\bse_j) = \delta_{ij}$ for a chosen basis
$\set{\bse_1, \dots, \bse_{16}}$ of $V$. Thus each $x_i$ acts as a
linear coordinate function on $V$, i.e., $x_i(\bsv) = v_i$ for any
$\bsv = \sum_{j=1}^{16} v_j \bse_j \in V$. A polynomial $f \in
\bbC[V] := \bbC[x_1, \dots, x_{16}]$ is called a
$\G$-\emph{invariant polynomial} if for every group element $(g,h)
\in \G$ and every $\bsv \in V$,
$$
f\Pa{\bsT_{(g,h)}^{-1} \bsv} = f(\bsv),
$$
or equivalently $(g,h)\cdot f = f$ under the action $[(g,h)\cdot
f](\bsv) := f(\bsT_{(g,h)}^{-1} \bsv)$. The set of all such
polynomials is precisely the invariant ring $\bbC[V]^{\G}$, which is
also called a $\bbC$-algebra.

Although $\G=\GL_A(2,\bbC)\times\GL_B(2,\bbC)$ is noncompact, its polynomial invariant ring on $V$ coincides with that of the compact group $\SU_A(2)\times\SU_B(2)$. This follows from two facts. First, scalar matrices act trivially by conjugation. Since every $g\in\GL(2,\bbC)$ can be written as $g=\lambda\bsS$, with $\lambda\in\bbC\setminus\{0\}$ and $\bsS\in\SL(2,\bbC)$, the conjugation actions of $\GL(2,\bbC)$ and $\SL(2,\bbC)$ have the same image. Second, $\SU(2)\times\SU(2)$ is Zariski dense in $\SL(2,\bbC)\times\SL(2,\bbC)$ \cite[Lemma~4.19]{Hoskins2015}. Since the conjugation action is a rational representation, these groups have the same complex polynomial invariant ring \cite[Proposition~1.3]{Turner2017}. Consequently,$$\bbC[V]^{\GL_A(2,\bbC)\times\GL_B(2,\bbC)}=\bbC[V]^{\SL_A(2,\bbC)\times\SL_B(2,\bbC)}=\bbC[V]^{\SU_A(2)\times\SU_B(2)}.$$ Moreover, local unitary phases cancel under conjugation, so $\U_A(2)\times\U_B(2)$ and $\SU_A(2)\times\SU_B(2)$ induce the same action on $V$. Thus, $\bbC[V]^{\U_A(2)\times\U_B(2)}=\bbC[V]^{\SU_A(2)\times\SU_B(2)}$. The compactness of $\SU_A(2)\times\SU_B(2)$ allows us to apply Molien's theorem to compute their common Molien series.

For the explicit construction of generators of $\bbC[V]^{\SU_A(2)\times\SU_B(2)}$, it is convenient, as in King \emph{et al.}, to exploit the local isomorphism between the group $\SU(2)$ and the corresponding rotation group $\SO(3)$. Let $V_{\bbR}:=\Vec\bigl(\operatorname{Herm}(\bbC^2\otimes\bbC^2)\bigr)\subset V$. Then $V_{\bbR}$ is a 16-dimensional real vector space, and $V\cong V_{\bbR}\otimes_{\bbR}\bbC$, where $V_{\bbR}\otimes_{\bbR}\bbC$ denotes the complexification of $V_{\bbR}$. The conjugate action of \(\SU_A(2)\times\SU_B(2)\) preserves \(V_{\mathbb{R}}\): for \((g,h)\in\SU_A(2)\times\SU_B(2)\), the induced map sends \(\Vec(\rho)\in V_{\mathbb{R}}\) to \(\Vec\bigl((g\otimes h)\rho(g\otimes h)^\dagger\bigr)\in V_{\mathbb{R}}\). The corresponding real polynomial invariant ring is $\bbR[V_{\bbR}]^{\SU_A(2)\times\SU_B(2)}$.

In the Bloch representation, the conjugation action of $\SU(2)$ on the three-dimensional real space spanned by the Pauli matrices factors through the homomorphism $\SU(2)\to\SO(3)$. The kernel of this homomorphism is $\{\pm\I\}$, and its image is the whole group $\SO(3)$. Equivalently, $\SU(2)$ is a double cover of $\SO(3)$ \cite{Hall2003}. Consequently, the local conjugation action of $\SU_A(2)\times\SU_B(2)$ on two-qubit operators induces an action on the real Pauli parameters, which factors through the orthogonal action of $\SO_A(3)\times\SO_B(3)$. For each $(g,h)\in\SU_A(2)\times\SU_B(2)$, let $\bsO_g\in\SO_A(3)$ and $\bsO_h\in\SO_B(3)$ denote the images of $g$ and $h$, respectively, under the double-covering homomorphism. The induced transformation of the Bloch parameters is \cite{Jing2016}
$$
t'=t,\qquad\bsa'=\bsO_g\bsa,\qquad\bsb'=\bsO_h\bsb,\qquad\bsC'=\bsO_g\bsC\bsO_h^\t.
$$
Since the homomorphism $\SU(2)\to\SO(3)$ is surjective, the two actions have the same orbits on $V_{\bbR}$ and hence the same real polynomial invariants: $$\bbR[V_{\bbR}]^{\SU_A(2)\times\SU_B(2)}=\bbR[V_{\bbR}]^{\SO_A(3)\times\SO_B(3)}.$$

Using the equality of the $\GL$ and $\SU$ invariant rings established	above and the identification $V\cong V_{\bbR}\otimes_{\bbR}\bbC$,
extension of scalars from $\bbR$ to $\bbC$ gives
$$
\bbC[V]^{\GL_A(2,\bbC)\times\GL_B(2,\bbC)}\cong\bbR[V_{\bbR}]^{\SU_A(2)\times\SU_B(2)}
\otimes_{\bbR}\bbC=\bbR[V_{\bbR}]^{\SO_A(3)\times\SO_B(3)}\otimes_{\bbR}\bbC.
$$
Therefore, the complex invariant ring $\bbC[V]^{\G}$ studied in this paper can be constructed conveniently by first determining the real $\SO_A(3)\times\SO_B(3)$-invariant polynomials in $t,\bsa,\bsb,\bsC$ and then complexifying the resulting real invariant ring. A generating set for the real invariant ring also generates the complex invariant ring over $\bbC$. In constructing this generating set, we take $K_1=t$, since $t$ is unchanged by the group action, and construct the remaining generators from $\bsa,\bsb,\bsC$. On the affine hyperplane of complex operators with trace one, $K_1=t=1/4$ is constant. The restrictions of the remaining 20 generators generate the complex polynomial invariant ring on this hyperplane under the same local conjugation action.

The following Hilbert finiteness theorem guarantees that $\bbC[V]^{\G}$ is finitely generated as a $\bbC$-algebra.

\begin{thrm}[\cite{Weyl1946,GoodmanWallach1998}]
For any reductive linear algebraic group $\sfG$ and any regular
finite-dimensional $\sfG$-module $\cX$, the ring $\bbC[\cX]^\sfG$ of
$\sfG$-invariant polynomials on $\cX$ is finitely generated by
homogeneous invariant polynomials.
\end{thrm}

In other words, there exist homogeneous invariant polynomials $p_1,
\ldots, p_l \in \bbC[\cX]^\G$ such that $\bbC[\cX]^\G = \bbC[p_1,
\ldots, p_l]$. This set $\set{p_1, \ldots, p_l}$ is called an
\emph{integrity basis} \cite{Quesne1976} for the ring of invariants.
An integrity basis is called \emph{fundamental} if none of its elements is redundant, and the corresponding polynomials are termed \emph{fundamental invariants}. Although integrity bases are non-unique, the minimal number of generators is uniquely determined by the group action.

Define $\bbC[\cX]^\G_m$ as the subspace spanned by all homogeneous
invariants of degree $m$. The invariant ring is therefore graded by
polynomial degree, satisfying
$\bbC[\cX]^\G=\oplus^\infty_{m=0}\bbC[\cX]^\G_m$. Let $n_m = \dim
\bbC[\cX]_m^\G$ denote the number of linearly independent
homogeneous polynomial invariants of degree $m$. These numbers can
be encoded into a generating function
$$
M_{\G}(q)= \sum_{m=0}^{\infty} n_m q^m,
$$
known as the \emph{Molien series}.

By the Hochster-Roberts theorem, for any reductive linear algebraic
group $\sfG$ and any regular finite-dimensional $\sfG$-module $\cX$,
the invariant ring $\bbC[\cX]^\sfG$ satisfies the \emph{Cohen-Macaulay
property} \cite{HochsterRoberts1974}.
\begin{thrm}[\cite{Sturmfels1993}]\label{CMP}
Let $\bbC[\cX]^\sfG$ be Cohen-Macaulay. There exists a
\emph{homogeneous system of parameters} $K_1,\ldots,K_n$ such that
$\bbC[\cX]^\G$ is finitely generated as a free module over
$\bbC[K_1,\ldots,K_n]$. Moreover, there exist homogeneous invariants
$\set{J_1,\ldots,J_r}$ with $J_0=1$ satisfying
\begin{eqnarray}\label{dsum}
\bbC[\cX]^\sfG =\bigoplus^r_{k=0}J_k\cdot \bbC[K_1,\ldots,K_n].
\end{eqnarray}
This direct-sum representation is referred to as the Hironaka
decomposition of $\bbC[\cX]^\sfG$.
\end{thrm}

The \emph{algebraically independent}\footnote{A set of elements is called algebraically independent if there is no nonzero polynomial that vanishes when these elements are substituted into it. Algebraic independence is the polynomial generalization of linear independence.} invariants $K_1,\ldots,K_n$ are
known as \emph{primary} invariants, and the \emph{linearly
independent} invariants $J_1,\ldots,J_r$ are known as
\emph{secondary} invariants. Theorem \ref{CMP} indicates that every
polynomial invariant $I \in \bbC[\cX]^\G$ admits a unique expansion $I
= \sum^r_{k=0}J_k f_k$, where each $f_k$ is a polynomial in the
primary invariants $K_1,\ldots,K_n$. For $\abs{q}<1$, this
decomposition yields an explicit closed form for the Molien series:
\begin{eqnarray}\label{MJK}
M_{\G}(q)=\frac{\sum_{k=0}^{r} q^{\deg J_k}}{\prod_{i=1}^n (1 - q^{\deg
K_i})}.
\end{eqnarray}
A full derivation is provided in Appendix~\ref{a}.

Related applications of Molien--Weyl methods, primary and secondary invariants, and Hironaka decomposition to multi-matrix invariant rings have recently been studied \cite{KochJevicki2025,KochKimVanZyl2026,KochRodrigues2026}.

\section{Calculation of the Molien series}\label{s3}

The Molien series serves not only as the generating function of the sequence $\set{n_m}_{m=0}^{\infty}$ but also, when expressed in the form associated with a Hironaka decomposition, encodes the numbers and degrees of the primary and secondary invariants.
It thus provides useful information about the invariant ring, although it does not by itself determine the explicit generators or their algebraic relations.

By the reduction established in Section~\ref{s2}, the invariant ring may be computed using the compact group \(\SU_A(2)\times\SU_B(2)\). Explicit computation of its Molien series relies on Molien's theorem, which we state below.
\begin{thrm}[\cite{Derksen2015}]
Let $\sfG$ be a compact Lie
group equipped with a normalized Haar measure $\dif\mu$ satisfying
$\int_\sfG \dif\mu(g) = 1$. Given a rational representation $\bsT :
\sfG \to \GL(V)$ with $\bsT(g) = \bsT_g \in \GL(V)$ for each $g \in
\sfG$, then the Molien series takes the integral form
$$
M_{\sfG}(q) = \int_{\sfG} \frac{\dif\mu(g)}{\det(\I_V - q \bsT_g)}.
$$
\end{thrm}

A detailed proof of this theorem is provided in
Appendix~\ref{Molien}.

On the compact subgroup $\G'=\SU_A(2)\times\SU_B(2)$,
the representation takes the form
$$
\bsT_{(g,h)}=g\ot\bar g\ot h\ot\bar h,
$$
where we have used $(g^{-1})^\t=\bar g$ and
$(h^{-1})^\t=\bar h$.
Since $\bbC[V]^\G=\bbC[V]^{\G'}$, the Molien series $M_\G(q)$ can be calculated as
\begin{eqnarray*}
M_\G(q) =M_{\G'}(q) = \int_{\G'}\frac{\dif\mu(g,h)}{\det(\I_{16}-q\bsT_{(g,h)})}.
\end{eqnarray*}

Recall that every element $g\in \SU_A(2)$ can be diagonalized into $\diag(e^{\mathrm{i}\frac{\alpha}2},
e^{-\mathrm{i}\frac{\alpha}2})$ for $\alpha\in[0,2\pi]$. For a class function $\varphi$ satisfying $\varphi(g)=\varphi(hgh^{-1})$ for all $h\in\SU_A(2)$, the Weyl integral formula for the normalized Haar measure on $\SU_A(2)$ yields
\begin{eqnarray*}
\int_{\SU(2)}\varphi(g)\dif\mu(g)=\frac{1}{|W|}\int_T
\varphi(t)f(t)\dif t= \int^{2\pi}_0
\varphi(\alpha)\Pa{\frac{1-\cos\alpha}{2\pi}\dif\alpha},
\end{eqnarray*}
where $W$ represents the Weyl group of $\SU_A(2)$,
$T=\set{\diag(e^{\mathrm{i}\frac{\alpha}2},
e^{-\mathrm{i}\frac{\alpha}2})~|~\alpha\in[0,4\pi]}$ is its maximal torus, and $f(t)=\Abs{e^{\mathrm{i}\frac{\alpha}2}-e^{-\mathrm{i}\frac{\alpha}2}}^2$. Although $\alpha$ ranges over $[0,4\pi]$ on the full maximal torus, restricting the integral to $[0,2\pi]$ by conjugacy symmetry introduces a factor of $2$, which cancels the Weyl-group factor $1/|W|=1/2$.

Suppose $h\in \SU_B(2)$ is diagonalized into $\diag(e^{\mathrm{i}\frac{\beta}2},
e^{-\mathrm{i}\frac{\beta}2})$. Then the normalized Haar measure $\dif\mu(g,h)$ over $\SU_A(2)\times\SU_B(2)$ can be written as
\begin{eqnarray*}
\dif\mu(g,h)&=&\dif\mu(g)\dif\mu(h)=\Pa{\frac{1-\cos\alpha}{2\pi}\dif\alpha}\Pa{\frac{1-\cos\beta}{2\pi}\dif\beta}\\
&=& \frac1{4\pi^2}(1-\cos\alpha)(1-\cos\beta)\dif\alpha\dif\beta,
\end{eqnarray*}
where $\alpha,\beta\in[0,2\pi]$. Now $\bsT_{(g,h)}$ is diagonalized
into
\begin{eqnarray*}
&&\diag(e^{\mathrm{i}\frac{\alpha}2},
e^{-\mathrm{i}\frac{\alpha}2})\ot
\diag(e^{-\mathrm{i}\frac{\alpha}2},
e^{\mathrm{i}\frac{\alpha}2})\ot \diag(e^{\mathrm{i}\frac{\beta}2},
e^{-\mathrm{i}\frac{\beta}2})\ot \diag(e^{-\mathrm{i}\frac{\beta}2},
e^{\mathrm{i}\frac{\beta}2})\\
&&=\diag(1,e^{\mathrm{i}\alpha},e^{-\mathrm{i}\alpha},1)\ot\diag(1,e^{\mathrm{i}\beta},e^{-\mathrm{i}\beta},1)
\end{eqnarray*}
whose full set of eigenvalues are
$\Set{1,e^{\pm\mathrm{i}\beta},1,e^{\mathrm{i}\alpha},e^{\mathrm{i}(\alpha\pm\beta)},e^{\pm\mathrm{i}\alpha},e^{\mathrm{i}(-\alpha\pm\beta)},e^{-\mathrm{i}\alpha},1,e^{\pm\mathrm{i}\beta},1}$.

Let $z=e^{\mathrm{i}\alpha}$ and $w=e^{\mathrm{i}\beta}$. Then $\dif
z=\mathrm{i}z\dif\alpha$ and $\dif w = \mathrm{i}w\dif\beta$, thus $\frac{\dif z\dif w}{\mathrm{i}z\cdot\mathrm{i}w}=\dif\alpha\dif\beta$. Since
$\abs{1-z}^2=(1-\cos\alpha)^2+\sin^2\alpha=2(1-\cos\alpha)$ and
$\abs{1-w}^2=2(1-\cos\beta)$, we get that
\begin{eqnarray*}
\det(\I_{16}-q\bsT_{(g,h)}) &=& (1-q)^4(1-qz)^2(1-q\bar
z)^2(1-qw)^2(1-q\bar w)^2\\
&&\times(1-qzw)(1-qz\bar w)(1-q\bar zw)(1-q\bar z\bar w)
\end{eqnarray*}
and
\begin{eqnarray*}
&&\frac1{4\pi^2}(1-\cos\alpha)(1-\cos\beta)=\frac1{(4\pi)^2}2(1-\cos\alpha)2(1-\cos\beta)\\
&&=\frac1{(4\pi)^2}(-\abs{1-z}^2)(-\abs{1-w}^2)=\frac1{(4\pi)^2}(1-z)^2(1-w)^2z^{-1}w^{-1}.
\end{eqnarray*}
Therefore
\begin{eqnarray}\label{M}
M_\G(q) &=&
\int_{\SU(2)}\int_{\SU(2)}{\det}^{-1}(\I_{16}-q\bsT_{(g,h)})\dif\mu(g)\dif\mu(h)\nonumber\\
&=&
\frac1{4\pi^2}\iint\frac{(1-\cos\alpha)(1-\cos\beta)\dif\alpha\dif\beta}{\det(\I_{16}-q\bsT_{(g,h)})}\nonumber\\
&=&\frac1{(4\pi\mathrm{i})^2}
\oint_{\sfU(1)}\oint_{\sfU(1)}\frac{(1-z)^2(1-w)^2z^{-2}w^{-2}}{(1-q)^4(1-qz)^2(1-q\bar
z)^2(1-qw)^2(1-q\bar w)^2}\nonumber\\
&&\times \frac{\dif z\dif w}{(1-qzw)(1-qz\bar w)(1-q\bar zw)(1-q\bar
z\bar w)}.
\end{eqnarray}

To compute the integral representation of the Molien series, we
employ complex residue calculus. We begin by recalling the
definition of the residue at an isolated singularity.
\begin{definition}[\cite{Marsden1987}]
Let $f$ be analytic on a punctured disk $0 < |z - z_0| < r$ with
Laurent expansion
$$
f(z) = \sum_{n=-\infty}^{\infty} a_n (z - z_0)^n.
$$
The \emph{residue} of $f$ at $z_0$ is defined as the coefficient
$a_{-1}$, denoted by $\Res_{z=z_0}f(z) = a_{-1}$.
\end{definition}

In practice, residues are often computed via contour integrals. By Cauchy's integral formula, the residue can also be expressed as
$$
\Res_{z=z_0}f(z)=\Res(f, z_0) = \frac{1}{2\pi \mathrm{i}}
\oint_{\abs{z - z_0} = \ell} f(z) \dif z,
$$
where $\ell > 0$ is small enough so that the circle encloses no
other singularities. The following standard result, known as the
Cauchy residue theorem, generalizes this formula to contours
enclosing multiple isolated singularities.

\begin{thrm}[\cite{Marsden1987}]\label{rest}
Let $f(z)$ be analytic inside and on a positively oriented simple
closed contour $C$, except for a finite number of isolated
singularities $z_1,z_2,\ldots,z_n$ inside $C$. Then
$$
\oint_{C}f(z)\dif z=2\pi\mathrm{i}\sum^n_{k=1}\Res_{z=z_k}f(z),
$$
where $\Res_{z=z_k}f(z)$ is the residue of $f$ at $z_k$.
\end{thrm}

The Residue Theorem reduces a contour integral around a closed curve
to the sum of residues at the isolated singularities inside the
curve. It is the fundamental tool for evaluating such integrals. For
higher-order poles, a more explicit formula is useful. Lemma
\ref{kpole} provides a convenient method to compute the residue when
$f(z)$ has a pole of order $k$ at $z_0$, expressed directly in terms
of the analytic part $\varphi(z)$.
\begin{lem}[\cite{Marsden1987}]\label{kpole}
Suppose that $\varphi(z)$ is analytic in $|z-z_0| < r$ and that $k
\geqslant 1$ is an integer. Let $f(z) =
\dfrac{\varphi(z)}{(z-z_0)^k}$. Then
$$
\Res_{z=z_0}f(z) = \frac{1}{(k-1)!} \lim_{z \to z_0}
\frac{\dif^{k-1}}{\dif z^{k-1}} \Br{(z-z_0)^k f(z)} =
\frac{\varphi^{(k-1)}(z_0)}{(k-1)!}.
$$
\end{lem}

Applying these results to the Molien integral in \eqref{M}, we
repeatedly invoke Theorem~\ref{rest} and Lemma~\ref{kpole}. The
final closed-form expression for $M_\G(q)$ is given in
Theorem~\ref{Mq}.

\begin{thrm}[\cite{Grassl1998,King2007}]\label{Mq}
For $\abs{q}<1$, it holds that
\begin{eqnarray*}
M_{\sfG}(q)&=&\frac{q^{10}-q^8-q^7+2 q^6+2 q^5+2 q^4-q^3-q^2+1}{(1-q)^{10} (1+q)^6(1+q^2)^2(1+q+q^2)^3}\\
&=&\frac{q^{15}+q^{11}+q^{10}+3 q^9+2 q^8+2 q^7+3
q^6+q^5+q^4+1}{(1-q)(1-q^2)^3(1-q^3)^2(1-q^4)^3(1-q^6)}.
\end{eqnarray*}
\end{thrm}

\begin{proof}
{\bf Step 1: Let us rewrite the integrand on the unit circles.} Note
that $\bar z=\frac1z$ and $\bar w=\frac1w$ since $z,w\in\sfU(1)$.
Hence $1-q\bar z=\frac{z-q}z,\quad 1-q\bar w=\frac{w-q}w$. Also
$1-q\bar z\bar w=\frac{zw-q}{zw},1-qz\bar w=\frac{w-qz}w, 1-q\bar
zw=\frac{z-qw}z$. Denote
$$
F_q(z,w):=\frac{z^2(1-z)^2w^2(1-w)^2}{(1-q)^4(1-qz)^2(z-q)^2(1-qw)^2(w-q)^2(1-qzw)(w-qz)(z-qw)(zw-q)}.
$$
Therefore the integral becomes
\begin{eqnarray*}
M_\G(q)&=&
\frac1{(4\pi\mathrm{i})^2}\oint_{\sfU(1)}\oint_{\sfU(1)}F_q(z,w)\dif
z\dif w.
\end{eqnarray*}
{\bf Step 2: We compute the $z$-integral by residues.} For fixed
$w\in\sfU(1)$, the rational function $F_q(z,w)$ in $z$ must have
three poles, determined by the zeros of its denominator.
\begin{itemize}
\item $z = q$ is a double pole arising from the factor $(z - q)^2$. For a double pole, using Lemma~\ref{kpole}, we can obtain
\begin{eqnarray*}
\Res_{z=q}F_q(z,w)&=&\partial_{z=q}\Br{(z-q)^2F_q(z,w)}\\
&=&\frac{w^2\Br{q^6w-2q^5w-q^4w+2q^3(w^2+1)-q^2w-2qw+w}}{(q-1)^5q(q+1)^3(q-w)^2(q^2-w)^2(qw-1)^2(q^2w-1)^2}.
\end{eqnarray*}
\item $z = q w$ is a simple pole arising from $(z - q w)$. Since $z=qw$ is a simple
pole arising from the factor $(z-qw)$, using Lemma \ref{kpole}, we
have
\begin{eqnarray*}
&&\Res_{z=qw}F_q(z,w)=\lim_{z\to qw}(z-qw)F_q(z,w)\\
&&=\lim_{z\to
qw}\frac{z^2(1-z)^2w^2(1-w)^2}{(1-q)^4(1-qz)^2(z-q)^2(1-qw)^2(w-q)^2(1-qzw)(w-qz)(zw-q)}\\
&&=\frac{(qw)^2(1-qw)^2w^2(1-w)^2}
{(1-q)^4(1-q^2w)^2(qw-q)^2(1-qw)^2(w-q)^2(1-q^2w^2)(w-q^2w)(qw^2-q)}\\
&&=\frac{w^3}{(q-1)^5q(q+1)(w^2-1)(q-w)^2(q^2w-1)^2(q^2w^2-1)}.
\end{eqnarray*}
\item $z = q / w$ is a simple pole arising from $(z w - q)$. Since $zw-q=0$ at
$z=q/w$, using Lemma \ref{kpole}, we have
\begin{eqnarray*}
&&\Res_{z=q/w}F_q(z,w)=\lim_{z\to q/w}\Pa{z-\frac qw}F_q(z,w)\\
&&=\lim_{z\to q/w}
\frac{z^2(1-z)^2w(1-w)^2}{(1-q)^4(1-qz)^2(z-q)^2(1-qw)^2(w-q)^2(1-qzw)(w-qz)(z-qw)}\\
&&=-\frac{w^3}{(q-1)^5q(q+1)(w^2-1)(q^2-w)^2(qw-1)^2(q^2-w^2)}.
\end{eqnarray*}
\end{itemize}
The other possible $\Set{\frac1q,\frac1{qw},\frac wq}$ poles are outside $\sfU(1)$ due to $\abs{q}<1$.
Hence
$$
\oint_{\sfU(1)}F_q(z,w)\dif z = 2\pi\mathrm{i}\Pa{\Res_{z=q}F_q(z,w)
+ \Res_{z=qw}F_q(z,w) + \Res_{z=q/w}F_q(z,w)}.
$$
After simplifying, one obtains
\begin{eqnarray*}
R_q(w)&:=&\Res_{z=q}F_q(z,w)+\Res_{z=qw}F_q(z,w)+\Res_{z=q/w}F_q(z,w)\\
&=&\frac{2w^2(w-1)^2\Pa{q^6w+q^4w+q^3(w+1)^2+q^2w+w}}{(q-1)^5(q+1)^3(q-w)^2(q^2-w)^2(q+w)(qw-1)^2(qw+1)(q^2w-1)^2}.
\end{eqnarray*}
Hence $\oint_{\sfU(1)}F_q(z,w)\dif z = 2\pi\mathrm{i}R_q(w)$,
and therefore
\begin{eqnarray*}
M_\G(q)&=& \frac1{8\pi\mathrm{i}}\oint_{\sfU(1)}R_q(w)\dif w.
\end{eqnarray*}
{\bf Step 3: In what follows, we compute the $w$-integral by
residues.} Now consider $R_q(w)$. There are three poles inside
$U(1)$.
\begin{itemize}
\item $w=q$ is a {double pole} from the factor $(q-w)^2$. Since the factor
$(q-w)^2=(w-q)^2$, we have
\begin{eqnarray*}
&&\Res_{w=q}R_q(w) = \partial_{w=q}\Br{(w-q)^2R_q(w)}\\
&&=\frac{-q^{12}+7q^{11}+13q^{10}+32q^9+51q^8+81q^7+82q^6+81q^5+51q^4+32q^3+13q^2+7q-1}{2(q-1)^{10}q(q+1)^6(q^2+1)^2(q^2+q+1)^3}
\end{eqnarray*}
\item $w=-q$ is a simple pole from the factor $(q+w)$. Since $q+w=0$ at
$w=-q$, it follows that
\begin{eqnarray*}
\Res_{w=-q}R_q(w) =
\lim_{w\to-q}(w+q)R_q(w)=\frac1{2(q-1)^6q(q+1)^4(q^2+1)^2}.
\end{eqnarray*}
\item $w=q^2$ is a double pole from the factor $(q^2-w)^2$. Since
$(q^2-w)^2=(w-q^2)^2$, it follows that
\begin{eqnarray*}
&&\Res_{w=q^2}R_q(w) = \partial_{w=q^2}\Br{(w-q^2)^2R_q(w)}\\
&&=-\frac{2q(q^2-q+1)\Pa{3q^4+6q^3+8q^2+6q+3}}{(q-1)^{10}(q+1)^4(q^2+1)^2(q^2+q+1)^3}.
\end{eqnarray*}
\end{itemize}
The other possible poles $\Set{\frac1q,-\frac1q,\frac1{q^2}}$ are outside $\sfU(1)$ due to $\abs{q}<1$.
Thus
\begin{eqnarray*}
\oint_{\sfU(1)}R_q(w)\dif w = 2\pi\mathrm{i}\Pa{\Res_{w=q}R_q(w) +
\Res_{w=-q}R_q(w) +\Res_{w=q^2}R_q(w)}.
\end{eqnarray*}
After carrying out these three computations and simplifying the
result, one gets
\begin{eqnarray*}
M_\G(q) &=& \frac14\Pa{\Res_{w=q}R_q(w) + \Res_{w=-q}R_q(w)+\Res_{w=q^2}R_q(w)}\\
&=&\frac{q^{10}-q^8-q^7+2 q^6+2 q^5+2
q^4-q^3-q^2+1}{(1-q)^{10}(1+q)^6(1+q^2)^2(1+q+q^2)^3}.
\end{eqnarray*}
This completes the proof.
\end{proof}

In the Bloch expansion, the $16$ parameters of the matrix
$\rho$ fall into four categories: one parameter $t$, three
parameters $a_i$, three parameters $b_j$, and nine parameters $c_{ij}$, where $i,j=1,2,3$. Under the action of the torus in
$\SU_A(2)\times\SU_B(2)$, the scalar parameter $t$
corresponds to the eigenvalue $1$. The three parameters
$a_i$ correspond to the eigenvalues $\{1,e^{\pm i\alpha}\}$,
while the three parameters $b_j$ correspond to
$\{1,e^{\pm i\beta}\}$. Since $c_{ij}$ carries one index
of each type, the nine corresponding eigenvalues are
the pairwise products of these two sets, namely
$\{1,e^{\pm i\alpha},e^{\pm i\beta},
e^{\pm i(\alpha+\beta)},e^{\pm i(\alpha-\beta)}\}$.
Each eigenvalue $\lambda$ contributes a factor
$1-q\lambda$ to $\det(\I-q\bsT_{g,h})$ in the Molien integral.

To distinguish the degrees contributed by the four categories,
replace $q$ by $t,a,b,c$, respectively, turning each factor into $1
- t, 1 - a e^{\pm i\alpha}$, etc. Multiplying all factors together
and performing the contour integration yields the refined Molien series $M_\G(t,a,b,c)$:
\begin{eqnarray*}
M_\G(t,a,b,c)&=&\frac{1}{(4\pi\mathrm{i})^2}\oint_{\sfU(1)}\oint_{\sfU(1)}
\frac{(1-z)^2z^{-2}(1-w)^2w^{-2}}
{(1-t)(1-az)(1-a)(1-a\bar z)(1-b w)(1-b)(1-b\bar w)}\\
&&\times\frac{\dif z\dif w}{(1-cwz)(1-c w)(1-cw\bar
z)(1-cz)(1-c)(1-c\bar z)(1-cz\bar w)(1-c\bar w)(1-c\bar w\bar z)}\\
&=&\frac{1}{(4\pi\mathrm{i})^2(1-t)(1-a)(1-b)(1-c)}\oint_{\sfU(1)}\oint_{\sfU(1)}\frac{(1-z)^2z^2(1-w)^2w^2}
{(1-az)(z-a)(1-b w)(w-b)}\\
&&\times\frac{\dif z\dif w}{(1-cwz)(1-c
w)(z-cw)(1-cz)(z-c)(w-cz)(w-c)(zw-c)}
\end{eqnarray*}
Denote by $F_{a,b,c}(z,w)$ the following expression
\begin{eqnarray*}
\frac{(1-z)^2z^2(1-w)^2w^2} {(1-az)(z-a)(1-b w)(w-b)(1-cwz)(1-cw)(z-cw)(1-cz)(z-c)(w-cz)(w-c)(zw-c)}.
\end{eqnarray*}
Then it holds that
\begin{eqnarray*}
M_\G(t,a,b,c) &=& \dfrac{\oint_{\sfU(1)}\dif
w\oint_{\sfU(1)}F_{a,b,c}(z,w)\dif z}{(4\pi
\mathrm{i})^2(1-t)(1-a)(1-b)(1-c)}.
\end{eqnarray*}

We first assume generic nonzero parameters with
$|t|,|a|,|b|,|c|<1$. For fixed $w\in\sfU(1)$, the poles of $F_{a,b,c}(z,w)$ in the
$z$-variable are $\{a,cw,c,c/w\}$, and all of them are simple poles. The other possible poles are outside $\sfU(1)$ due to the fact
$\abs{a}<1$ and $\abs{c}<1$. Hence
\begin{eqnarray*}
\oint_{\sfU(1)}F_{a,b,c}(z,w)\dif z =
2\pi\mathrm{i}\sum_{z\in\set{a,cw,c,c/w}}\Res(F_{a,b,c},z).
\end{eqnarray*}
By calculation, the sum of the four $z$-residues is given by {\scriptsize
$$
R_{a,b,c}(w):=-\frac{2 (w-1)^2 w^2 (a^2 c^5 w-a^2 c^4 w+a^2 c^3 w+a
c^3 w^2+a c^3 w+a c^3-a c^2 w^2-a c^2 w-a c^2-c^2 w+c w-w)}{(a+1)
(c-1) (c+1) (a c-1) (w-b) (b w-1) (c-w) (c^2-w) (c+w) (c w-1) (c
w+1) (c^2 w-1) (a c-w) (a c w-1)}.
$$}
Then
\begin{eqnarray}\label{mg}
M_\G(t,a,b,c)=\frac{\oint_{\sfU(1)}R_{a,b,c}(w)\dif
w}{8\pi\mathrm{i}(1-t)(1-a)(1-b)(1-c)}.
\end{eqnarray}

Now consider $R_{a,b,c}(w)$, its poles inside $\sfU(1)$ are
$\set{c,c^2,b,ac,-c}$, and all of them are simple poles. The other
possible poles are outside $\sfU(1)$ due to the fact that
$|a|<1$, $|b|<1$, and $|c|<1$. Thus
\begin{eqnarray*}
\oint_{\sfU(1)}R_{a,b,c}(w)\dif w =
2\pi\mathrm{i}\sum_{w\in\set{c,c^2,b,ac,-c}}\Res(R_{a,b,c},w).
\end{eqnarray*}

{\scriptsize
$$
$$}
By calculation, $\oint_{\sfU(1)}R_{a,b,c}(w)\dif w $ is given by
{\begin{align*}
\frac{8\pi\mathrm{i}N(a,b,c)}{(1+a)(1+b)(1-c^2)(1-c^4)(c^2+c+1)(1-a^2c^2)(1-ac^2)(1-b^2c^2)(1-bc^2)(1-abc)}
\end{align*}}
where
{\begin{eqnarray*}
N(a,b,c)&=&1-c^2\Big[a^3b^3c^7	-a^2b^2c^5(a+b-1)+abc^3(a^2+ab+b^2)
-c^2(a^2+ab+b^2)\\
&&+abc^4(ab+a+b) -abc(a+b+1)-ab+a+b	\Big].
\end{eqnarray*}}

Substituting this result into the contour-integral formula (\ref{mg}) and simplifying the resulting rational function, we obtain the refined Molien series $$M_\G(t,a,b,c)=\frac{N(a,b,c)}{D(t,a,b,c)(1-ac^2)(1-bc^2)},$$
where, for convenience, we have defined
\begin{align*}
D(t,a,b,c)
=(1-t)(1-a^2)(1-b^2)(1-c^2)(1-c^4)(1-c^3)(1-a^2c^2)(1-b^2c^2)(1-abc).
\end{align*}
The resulting expression extends by continuity to parameter
values where poles coincide or parameters vanish.

To display its structure more clearly, define
\begin{align*}
N_0(a,b,c)
&=1+abc^2+abc^3+a^2b^2c^5,\\
N_1(a,b,c)
&=ab^2c^3+a^2bc^4+a^2bc^5+a^3c^6,\\
N_2(a,b,c)
&=a^2bc^3+ab^2c^4+ab^2c^5+b^3c^6.
\end{align*}
A direct algebraic rearrangement gives
\begin{eqnarray*}
N(a,b,c)(1+ac^2)(1+bc^2)
&=&N_0(a,b,c)\Big[(1-a^2c^4)(1-b^2c^4)+a^2c^4(1-b^2c^4)+b^2c^4(1-a^2c^4)\Big]\\
&&+N_1(a,b,c)(1-b^2c^4)	+N_2(a,b,c)(1-a^2c^4).
\end{eqnarray*}
Consequently, the refined Molien series takes the form
\begin{align*}
M_\G(t,a,b,c)=\frac{1}{D(t,a,b,c)}\Big[N_0(a,b,c)	\left(	1+\frac{a^2c^4}{1-a^2c^4}
+\frac{b^2c^4}{1-b^2c^4}\right)+\frac{N_1(a,b,c)}{1-a^2c^4}
+\frac{N_2(a,b,c)}{1-b^2c^4}\Big].
\end{align*}
Under the identification $(a,s,p,b)_{\mathrm{King}}=(t,a,b,c)_{\mathrm{ours}}$, this expression coincides precisely with Eq.~(3.7) of Ref.~\cite{King2007}. As noted therein, this form makes the non-negativity of the coefficients in the corresponding multigraded power-series expansion manifest.

Setting \(t=a=b=c=q\), the refined Molien series reduces to
\begin{eqnarray*}
M_{\G}(q)= \frac{q^{10}-q^8-q^7+2 q^6+2 q^5+2
q^4-q^3-q^2+1}{(q-1)^{10} (q+1)^6 (q^2+1)^2 (q^2+q+1)^3}.
\end{eqnarray*}

\section{An integrity basis constructed by the graphical method}\label{s4}

In the previous section, the closed-form expression of the Molien
series for the ring of local polynomial invariants of two-qubit
mixed states was obtained via contour integration based on Molien's
theorem. This series not only provides the dimension sequence of
homogeneous invariants in each degree, but also encodes the degree
information of primary and secondary invariants, thereby supplying
essential information for a complete characterization of the
structure of the invariant ring. However, to explicitly construct a
set of generators and clarify the algebraic relations among them, a
systematic method is still required.

In this section, we use the graphical tensor method of King \emph{et al.}, which
transforms the tensor contraction process of the Bloch parameters
$a_\alpha$, $b_i$ and $c_{\alpha i}$ into the enumeration of
connected graphs composed of elementary building blocks. This method
exploits the invariant tensors $\delta_{ij}$ and $\varepsilon_{ijk}$
of $\mathrm{SO}(3)$ to classify all possible complete contractions,
thereby avoiding the tediousness and omissions inherent in purely
algebraic enumeration. In Subsection~\ref{graphical_rules} below, we
first present the basic rules and building blocks of the graphical
method; then in Subsection~\ref{integrity-basis}, we enumerate all
connected graphs case by case, screen and reduce the candidate
invariants, and finally obtain an integrity basis consisting of 21
generators.

\subsection{Basic rules and building blocks of the graphical method}\label{graphical_rules}

As established in Section~2, we construct the generators using the $\SO_A(3)\times\SO_B(3)$ action on the real Bloch parameters $\bsa,\bsb,\bsC$. Throughout the remainder of this section, we use Dirac-type notation
for scalar products and tensor contractions. For real vectors
$\bsu,\bsv\in\mathbb{R}^{3}$, we write
$\inner{\bsu}{\bsv}=\bsu^{\t}\bsv$. For real matrices
$\bsX,\bsY\in\mathbb{R}^{3\times 3}$, the same notation denotes the Hilbert-Schmidt inner product, $\inner{\bsX}{\bsY}=\Tr{\bsX^{\t}\bsY}$. Moreover,
$\innerm{\bsu}{\bsM}{\bsv}=\bsu^{\t}\bsM\bsv$ denotes the complete
contraction of two vectors with a matrix.

Since any invariant must be a scalar, it cannot possess any free
indices. Therefore, any candidate invariant constructed from
$\bsa,\bsb$, and $\bsC$ must have all its indices fully contracted.
For example, expressions such as $\inner{\bsa}{\bsa}$ and
$\inner{\bsb}{\bsb}$ are invariant under independent $\SO(3)$
rotations of each qubit, because the rotation matrices cancel
pairwise via $\sum_j \bsO_{ij} (\bsO^\t)_{jk} = \delta_{ik}$.
Similarly, for cross contractions such as
$\innerm{\bsa}{\bsC}{\bsb}$, the rotation matrices acting on the
$A$-space and the $B$-space cancel independently via the same
orthogonality relation.

On the other hand, any term with uncontracted indices, such as an
isolated $a_\alpha$ or $c_{\alpha i}$, transforms nontrivially under
rotations and hence cannot itself be a scalar invariant. However,
this does not prevent these tensors from appearing in invariant
expressions, provided that all indices are fully contracted, as in
$\Inner{\bsa}{\bsa}=\sum_{\alpha} a_\alpha a_\alpha$ or
$\innerm{\bsa}{\bsC}{\bsb}=\sum_{\alpha,i}a_\alpha c_{\alpha i}
b_i$. This principle reduces the problem
of finding all polynomial invariants to that of enumerating all
possible complete contractions of the tensors $a_\alpha, b_i$, and
$c_{\alpha i}$.

As noted in \cite{Almumin2020}, it is known that for the group
$\SO(3)$, the invariant tensors are $\delta_{ij}$ and
$\varepsilon_{ijk}$. These two tensors generate all polynomial
invariants of the fundamental representation of $\SO(3)$. In the
context of the two-qubit system, this implies that any $\G$-invariant
polynomial constructed from the Bloch vectors $\bsa, \bsb$, and the
correlation matrix $\bsC$ can be obtained by taking all possible
complete contractions of indices using $\delta_{ij}$ and
$\varepsilon_{ijk}$.

The discussion above reduces the problem of finding all polynomial
invariants to that of enumerating all possible complete contractions
of the tensors $t, a_\alpha, b_i, c_{\alpha i}$, and
$c_{i\alpha}^\t$, using $\delta^A_{\alpha\beta}, \delta^B_{ij},
\varepsilon^A_{\alpha\beta\gamma}$, and $\varepsilon^B_{ijk}$. As
the number of independent contractions grows rapidly with the
polynomial degree, a purely algebraic enumeration becomes cumbersome
and prone to omission. It is therefore convenient to adopt a
graphical representation of tensor contractions, in which each
invariant corresponds to a diagram built from elementary vertices
and index contractions. This method has been applied in \cite{Welsh}
and \cite{Penrose}. The basic building blocks used in constructing
the invariants are shown in Table~\ref{building_blocks}.

\begin{table}[htbp]
\centering
\caption{The building blocks for the graphical
construction of invariants.}
\label{building_blocks}
\begin{tabular}{cc p{9cm}}
\toprule
\textbf{Symbol} & \textbf{Graphical representation} & \textbf{Meaning} \\
\midrule
$t$ & isolated vertex & scalar, no indices \\
$a_\alpha$ & \raisebox{-0.4\height}{\includegraphics[trim=14 22 9 22, clip, width=1.1cm]{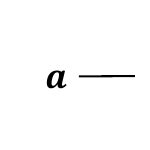}} & first-order tensor, carrying an $\SO_A(3)$ index $\alpha$. \\
$b_i$ &\raisebox{-0.4\height}{\includegraphics[trim=14 22 9 22, clip, width=1.1cm]{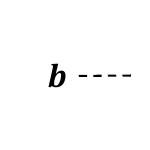}} & first-order tensor, carrying an $\SO_B(3)$ index $i$. \\
$c_{\alpha i}$ or $c^\t_{i\alpha}$&\raisebox{-0.4\height}{\includegraphics[trim=0 22 0 22, clip, width=1.6cm]{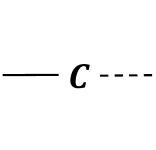}} & second-order tensor, carrying an $\SO_A(3)$ index $\alpha$ and an $\SO_B(3)$ index $i$.\\
$\varepsilon^A_{\alpha\beta\gamma}$ &\raisebox{-0.75\height}{\includegraphics[trim=0 0 0 31, clip, width=1.6cm]{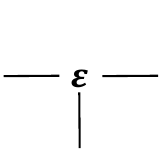}} & third-order antisymmetric tensor, carrying three $\SO_A(3)$ indices. \\
$\varepsilon^B_{ijk}$ &\raisebox{-0.75\height}{\includegraphics[trim=0 0 0 31, clip, width=1.6cm]{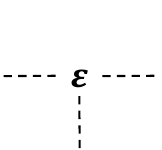}} & third-order antisymmetric tensor, carrying three $\SO_B(3)$ indices.\\
\bottomrule
\end{tabular}
\end{table}

In the graphical method, index contractions are represented by
connecting lines. Once two legs carrying the same type of $\SO(3)$
indices are connected, this is equivalent to a contraction with the
Kronecker delta $\delta^A_{\alpha\beta}$ or $\delta^B_{ij}$. Hence
the tensor need not appear as an independent building block. In
contrast, the third-order Levi-Civita tensors
$\varepsilon^A_{\alpha\beta\gamma}$ and $\varepsilon^B_{ijk}$ cannot
be represented by connecting lines alone. They must therefore appear
explicitly as building blocks in the graphs.

The invariants are constructed by connecting the legs of the basic
building blocks pairwise until no free edges remain in the connected
graph. Each connecting line represents a contraction of a pair of
indices of the same $\SO(3)$ type. Solid lines connect legs carrying
$\SO_A(3)$ indices, while dashed lines connect legs carrying
$\SO_B(3)$ indices. When translating graphs into algebraic
expressions, the repeated index convention is adopted: any index
appearing twice in the same term is summed over its full range of
values, so that each connecting line corresponds to a summation. A
complete graph with no free indices represents a scalar polynomial,
which is an $\SO(3) \times \SO(3)$ invariant. All polynomial
invariants can be obtained by enumerating all possible connection
patterns.

\subsection{The complete set of generators}
\label{integrity-basis}

To obtain a finite set of generators for the invariant ring, in this
section we construct candidate invariants by enumerating all
possible connected graphs. In translating diagrams into algebraic
expressions, we adopt the notation from \cite{Zhang2025} for the
invariants. In particular, we write
\begin{eqnarray}\label{aFbF}
\bsa \cdot \cF^A = \sum_{\alpha=1}^3 a_\alpha \bsF_\alpha^A, \quad
\bsb \cdot \cF^B = \sum_{i=1}^3 b_i \bsF_i^B,
\end{eqnarray}
where $\cF^A = (\bsF_1^A, \bsF_2^A, \bsF_3^A)$ and $\cF^B =
(\bsF_1^B, \bsF_2^B, \bsF_3^B)$ are vectors of matrices
$\bsF^A_\alpha:=(\varepsilon_{\alpha\beta\gamma}^A)_{\beta,\gamma}$
and $\bsF_i^B:=(\varepsilon_{ijk}^B)_{j,k}$, respectively. It holds
that $(\bsa\cdot\cF)^\t=-\bsa\cdot\cF$. The corresponding graphical
representation for Eq.~\eqref{aFbF} is shown in Figure~\ref{aF}.

We also denote by $\bsM^*$ the \emph{adjugate matrix} of any matrix
$\bsM$, and set $\widehat{\bsM}:=(\bsM^*)^\t$, which satisfies $\bsM
\widehat{\bsM}^\t = \widehat{\bsM}^\t \bsM = \det(\bsM)\I_3$. In
Linear Algebra, for any two square matrices $\bsM$ and $\bsN$ of
order $n$, it is well known that
$\widehat{\bsM^\t}=\widehat{\bsM}^\t$ and
$\widehat{\bsM\bsN}=\widehat{\bsM}\widehat{\bsN}$.

\begin{figure}[H] \centering
\includegraphics[trim=5 10 1 18, clip, width=0.4\textwidth]{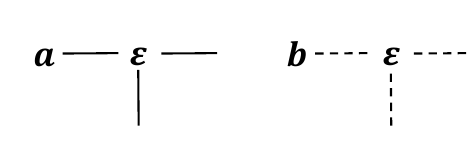}
\caption{The graphical representation of $\bsa\cdot\cF^A$ and
$\bsb\cdot\cF^B$.} \label{aF}
\end{figure}

Since every complete contraction involving $t$ factors as a power of $t$ times a contraction independent of $t$, the only generator
involving $t$ that needs to be retained is $t$ itself. We therefore consider the remaining five building blocks. Denote $\Omega:=\Set{\bsa,\bsb,\bsC,\varepsilon^A,\varepsilon^B}$. According to whether each building-block type in $\Omega$ occurs, there are $31$ possible cases, which we list below. Each case specifies which types of building blocks occur, but does not fix their numbers or connection patterns. The reductions below allow us to examine a finite set of candidate graphs within each case.

\begin{itemize}
\item Those containing a single building-block type in $\Omega$ (there are $\binom{5}{1}$ possibilities): (1) $\bsa$, (2) $\bsb$, (3) $\bsC$, (4) $\varepsilon^A$, (5) $\varepsilon^B$;
\item Those containing exactly two tensors in $\Omega$ (there are $\binom{5}{2}$ possibilities): (6) $\bsa, \bsb$, (7) $\bsa, \bsC$, (8) $\bsb,\bsC$, (9) $\bsa, \varepsilon^A$, (10) $\bsa, \varepsilon^B$, (11) $\bsb, \varepsilon^A$, (12) $\bsb, \varepsilon^B$, (13) $\bsC, \varepsilon^A$, (14) $\bsC, \varepsilon^B$, (15) $\varepsilon^A, \varepsilon^B$;
\item Those containing exactly three tensors in $\Omega$ (there are $\binom{5}{3}$ possibilities): (16) $\bsa, \bsb, \bsC$, (17) $\bsC, \varepsilon^A, \varepsilon^B$, (18) $\bsa, \varepsilon^A, \varepsilon^B$, (19) $\bsb, \varepsilon^A, \varepsilon^B$, (20) $\bsa,\bsb, \varepsilon^A$, (21) $\bsa, \bsb, \varepsilon^B$, (22) $\bsb, \bsC, \varepsilon^A$, (23) $\bsa, \bsC, \varepsilon^B$, (24) $\bsb, \bsC, \varepsilon^B$, (25) $\bsa, \bsC, \varepsilon^A$;
\item Those containing exactly four tensors in $\Omega$ (there are $\binom{5}{4}$ possibilities): (26) $\bsa, \bsb, \bsC, \varepsilon^A$, (27) $\bsa, \bsb, \bsC, \varepsilon^B$, (28) $\bsa, \bsb, \varepsilon^A, \varepsilon^B$, (29) $\bsa, \bsC, \varepsilon^A, \varepsilon^B$, (30) $\bsb, \bsC, \varepsilon^A, \varepsilon^B$;
\item Those containing all five tensors in $\Omega$ (there are $\binom{5}{5}$ possibilities): (31) $\bsa, \bsb, \bsC, \varepsilon^A, \varepsilon^B$.
\end{itemize}
For clarity, we summarize the above 31 possible cases in Table \ref{cases}.

\begin{table}[H]
\centering
\caption{$31$ potential possibilities.}
\begin{tabular}{c|c|c|c|c}
\hline
only one type & only two types & only three types & only four types & all five types \\
\hline
(1) $\bsa$ & (6) $\bsa, \bsb$& (16) $\bsa, \bsb, \bsC$ &(26) $\bsa, \bsb, \bsC, \varepsilon^A$& (31) $\bsa, \bsb, \bsC, \varepsilon^A, \varepsilon^B$\\
\hline
(2) $\bsb$ &(7) $\bsa, \bsC$ &(17) $\bsC, \varepsilon^A, \varepsilon^B$ & (27) $\bsa, \bsb, \bsC, \varepsilon^B$&\\
\hline
(3) $\bsC$ & (8) $\bsb,\bsC$ &(18) $\bsa, \varepsilon^A, \varepsilon^B$&(28) $\bsa, \bsb, \varepsilon^A, \varepsilon^B$& \\
\hline
(4) $\varepsilon^A$ & (9) $\bsa, \varepsilon^A$ & (19) $\bsb, \varepsilon^A, \varepsilon^B$ &(29) $\bsa, \bsC, \varepsilon^A, \varepsilon^B$&\\
\hline
(5) $\varepsilon^B$ & (10) $\bsa, \varepsilon^B$ & (20) $\bsa,\bsb, \varepsilon^A$& (30) $\bsb, \bsC, \varepsilon^A, \varepsilon^B$& \\
\hline
& (11) $\bsb, \varepsilon^A$ & (21) $\bsa, \bsb, \varepsilon^B$& &\\
\hline
& (12) $\bsb, \varepsilon^B$ & (22) $\bsb, \bsC, \varepsilon^A$ & &\\
\hline
& (13) $\bsC, \varepsilon^A$ & (23) $\bsa, \bsC, \varepsilon^B$& &\\
\hline
& (14) $\bsC, \varepsilon^B$ & (24) $\bsb, \bsC, \varepsilon^B$& &\\
\hline
& (15) $\varepsilon^A, \varepsilon^B$ & (25) $\bsa, \bsC, \varepsilon^A$& &\\
\hline
\end{tabular}\label{cases}
\end{table}

To obtain invariants by tensor contraction, $A$-type indices can only
contract with $A$-type indices, and $B$-type indices only with
$B$-type indices . The cases (6), (10), (11), (15), (18), (19), (20),
(21), and (28) are excluded because they contain both $A$-type and $B$-type indices but lack $\bsC$ as a mediator to form a connected
graph\footnote{A disconnected graph can be decomposed into multiple connected components, each of which is itself a connected subgraph. Hence a disconnected graph corresponds to a product of lower-degree invariants and is redundant as a generator.}. Among the
remaining cases, further exclusions are made when the corresponding
invariant either vanishes identically or reduces to a polynomial
combination of lower-degree invariants.

For convenience, vectors transforming under the two different groups
$\SO_A(3)$ and $\SO_B(3)$ are referred to as $A$-type vectors and
$B$-type vectors, respectively. Here, $A$-type indices are denoted
by $\alpha,\beta,\gamma,\dots$, while $B$-type indices are denoted
by $i,j,k,\dots$. Thus, an $A$-type vector $\bsA$ is a quantity
carrying one free $A$-type index, taking the form $a_\alpha, b_i
c_{i\alpha}^\t, a_\beta c_{\beta i} c_{i\alpha}^\t, \ldots$, where
$\bsa$ and $\bsb$ are contracted with products of even and odd
numbers of matrices $\bsC$ or $\bsC^\t$, respectively. Similarly, a
$B$-type vector $\bsB$ is a quantity carrying one free $B$-type
index, taking the form $b_i, a_\alpha c_{\alpha i}, b_j
c_{j\alpha}^\t c_{\alpha i}, \ldots$, where $\bsa$ and $\bsb$ are
contracted with products of odd and even numbers of matrices $\bsC$
or $\bsC^\t$, respectively. The graphical representations of vectors
of $A$-type and $B$-type constructed from $\bsa, \bsb, \bsC$ are
 shown in Figure~\ref{AB}.

\begin{figure}[htbp] \centering
\includegraphics[trim=5 9 7 10, clip, width=0.75\textwidth]{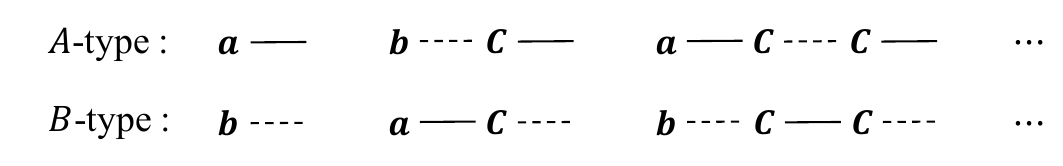}
\caption{The graphical representation of some $A$- and $B$-type
vectors.} \label{AB}
\end{figure}

To eliminate those candidates that vanish identically or reduce to
products of lower-degree invariants, we need several reduction
identities. The following proposition provides a key identity,
showing that any chain-like subgraph containing more than four
$\bsC$-nodes can be reduced to a combination of shorter chains. This
fact will significantly reduce the number of candidates in the
subsequent enumeration.

\begin{prop}\label{4c}
For any $\bsM\in\bbR^{3\times 3}$, it holds that
\begin{eqnarray*}
\bsM\bsM^\t\bsM\bsM^\t\bsM=\Inner{\bsM}{\bsM} \bsM\bsM^\t\bsM
+\det(\bsM)\widehat\bsM -
\tfrac12\Br{\Inner{\bsM}{\bsM}^2-\Inner{\bsM^\t\bsM}{\bsM^\t\bsM}}\bsM.
\end{eqnarray*}
\end{prop}

\begin{proof}
Recall from \cite[Corollary B1]{Zhang2025} that
$(\bsM^\t\bsM)^2=\Inner{\bsM}{\bsM}\bsM^\t\bsM+\widehat{\bsM^\t\bsM}
- \Inner{\widehat\bsM}{\widehat\bsM}\I_3$. Multiplying on the left
by $\bsM$ gives
\begin{eqnarray*}
\bsM\bsM^\t\bsM\bsM^\t\bsM &=& \Inner{\bsM}{\bsM}\bsM\bsM^\t\bsM
+\bsM\widehat{\bsM^\t\bsM}- \Inner{\widehat\bsM}{\widehat\bsM}\bsM\\
&=&\Inner{\bsM}{\bsM}\bsM\bsM^\t\bsM +\det(\bsM)\widehat{\bsM} -
\Inner{\widehat\bsM}{\widehat\bsM}\bsM.
\end{eqnarray*}
By the same corollary,
$$
\Inner{\widehat\bsM}{\widehat\bsM} =
\frac12\Br{\Inner{\bsM}{\bsM}^2-\Inner{\bsM^\t\bsM}{\bsM^\t\bsM}}.
$$
The desired identity follows.
\end{proof}

The Cayley--Hamilton identity satisfied by $\bsM\bsM^\t$ follows directly from Proposition~\ref{4c}. Indeed, right-multiplying the identity in Proposition~\ref{4c} by $\bsM^\t$ and using $\widehat{\bsM}\bsM^\t=\det(\bsM)\I_3$ yields the following corollary.

\begin{cor}\label{CH}
For any $\bsM\in\bbR^{3\times 3}$, it holds that
\begin{eqnarray*}
(\bsM\bsM^\t)^3=\Inner{\bsM}{\bsM} (\bsM\bsM^\t)^2-
\tfrac12\Br{\Inner{\bsM}{\bsM}^2-\Inner{\bsM^\t\bsM}{\bsM^\t\bsM}}\bsM\bsM^\t
+\det(\bsM)^2\I_3 .
\end{eqnarray*}
\end{cor}

In the diagrammatic approach, certain contractions involving
subgraphs of cross products can be further simplified. Specifically,
the following proposition shows that the action of the matrix $\bsM$
on the cross product $\bsu \times \bsv$ can be decomposed into three
terms, each involving only standard matrix-vector multiplications.

\begin{prop}\label{p43}
For any $\bsM\in\bbR^{3\times3}$ and $\bsu,\bsv\in\bbR^3$, it holds
that
\begin{eqnarray*}
\bsM(\bsu\times\bsv) = \Tr{\bsM}(\bsu\times\bsv) -
(\bsM^\t\bsu)\times \bsv - \bsu\times(\bsM^\t\bsv).
\end{eqnarray*}
Moreover $\bsM\bsu\times\bsM\bsv = \widehat\bsM(\bsu\times\bsv)$.
\end{prop}

For this proposition, there are two different methods of proof.
\begin{proof}[The first proof]
Recall that, for any $\bsx,\bsy,\bsz\in\bbR^3$, we see that
$\Inner{\bsx\times\bsy}{\bsz} = \det\Pa{\bsx,\bsy,\bsz}$. Note that
\begin{itemize}
\item $\Inner{\bsw}{\bsM(\bsu\times\bsv)}=\Inner{\bsM^\t\bsw}{\bsu\times\bsv}=\Inner{\bsu\times\bsv}{\bsM^\t\bsw}=\det(\bsu,\bsv,\bsM^\t\bsw)$,
\item $\Inner{\bsw}{(\bsM^\t\bsu)\times \bsv}=\Inner{\bsv\times\bsw}{\bsM^\t\bsu}=\Inner{\bsM^\t\bsu}{\bsv\times\bsw}=\det(\bsM^\t\bsu,\bsv,\bsw)$,
\item $\Inner{\bsw}{\bsu\times(\bsM^\t\bsv)}=\Inner{\bsu\times(\bsM^\t\bsv)}{\bsw}=\det(\bsu,\bsM^\t\bsv,\bsw)$.
\end{itemize}
Summing three dot products up, we get that
\begin{eqnarray*}
&&\Inner{\bsw}{\bsM(\bsu\times\bsv)+(\bsM^\t\bsu)\times
\bsv+\bsu\times(\bsM^\t\bsv)}\\
&&=\det(\bsM^\t\bsu,\bsv,\bsw)+\det(\bsu,\bsM^\t\bsv,\bsw)+\det(\bsu,\bsv,\bsM^\t\bsw)\\
&&=\Tr{\bsM}\det(\bsu,\bsv,\bsw)=\Tr{\bsM}\Inner{\bsw}{\bsu\times\bsv}
\end{eqnarray*}
holds for any $\bsw\in\bbR^3$. This leads to the conclusion that
$$
\bsM(\bsu\times\bsv)+(\bsM^\t\bsu)\times
\bsv+\bsu\times(\bsM^\t\bsv)=\Tr{\bsM}\bsu\times\bsv.
$$
For the second item, we choose any $\bsw\in\bbR^3$, and consider the
following calculation,
\begin{eqnarray*}
\Inner{\bsM\bsw}{\bsM\bsu\times\bsM\bsv}&=&
\det(\bsM\bsu,\bsM\bsv,\bsM\bsw) =
\det(\bsM(\bsu,\bsv,\bsw))\\
&=&\det(\bsM)\det(\bsu,\bsv,\bsw)=\Inner{\det(\bsM)\bsw}{\bsu\times\bsv}\\
&=&
\Inner{\widehat\bsM^\t\bsM\bsw}{\bsu\times\bsv}=\Inner{\bsM\bsw}{\widehat\bsM(\bsu\times\bsv)},
\end{eqnarray*}
implying that $\bsM\bsu\times\bsM\bsv=\widehat\bsM(\bsu\times\bsv)$
if $\bsM$ is invertible. When $\bsM$ is not invertible, we can find
a net of invertible matrices $\set{\bsM_\epsilon}_\epsilon$ such
that $\bsM=\lim_{\epsilon\to0}\bsM_\epsilon$. Thus
$$
\bsM_\epsilon\bsu\times\bsM_\epsilon\bsv =
\widehat\bsM_\epsilon(\bsu\times\bsv).
$$
By taking the limit with respect to $\epsilon$ approaching $0$, we
get that $\bsM\bsu\times\bsM\bsv=\widehat{\bsM}(\bsu\times\bsv)$.
\end{proof}

\begin{proof}[The second proof]
Using \cite[Proposition B1]{Zhang2025}, we can get that
$\bsM^*=\widehat{\bsM}^\t=\bsM^2-\tr{\bsM}\bsM+\tr{\widehat{\bsM}}\I_3$.
By taking the traces on both sides, we get that
$\tr{\widehat{\bsM}}=\frac12\Pa{\tr{\bsM}^2-\tr{\bsM^2}}$, which
implies
\begin{eqnarray}\label{M*}
\bsM^*=\bsM^2-\tr{\bsM}\bsM+
\frac12\Pa{\tr{\bsM}^2-\tr{\bsM^2}}\I_3.
\end{eqnarray}
For $\bsM+\I_3$, it holds that
\begin{eqnarray}\label{M+I}
(\bsM + \I_3)^*=(\bsM + \I_3)^2-\tr{(\bsM + \I_3)}(\bsM + \I_3)+
\frac12\Pa{\tr{(\bsM + \I_3)}^2-\tr{(\bsM + \mathbb{1}_3)^2}}\I_3 .
\end{eqnarray}
Subtracting Eq.~\eqref{M*} and $\I_3^*=\I_3$ from Eq.~\eqref{M+I},
we obtain $(\bsM + \I_3)^* - \bsM^* - \I_3^* = \Tr{\bsM} \I_3 -
\bsM$.

Recall from \cite[Corollaries B5 and B6]{Zhang2025} that for any
$\bsM,\bsN\in\bbR^{3\times 3}$ and $\bsu,\bsv\in\bbR^3$,
$$
\bsM\bsu \times \bsN\bsv + \bsN\bsu \times \bsM\bsv =
\Pa{\widehat{\bsM + \bsN} - \widehat{\bsM} - \widehat{\bsN}} (\bsu
\times \bsv).
$$
Applying this identity with \(M\) replaced by \(\bsM^{\mathsf \T}\) and \(\bsN=\I_3\), we obtain
\begin{eqnarray*}
\bsM^\t \bsu \times \bsv + \bsu \times \bsM^\t \bsv
&=& \Pa{\widehat{\bsM + \I_3}^\t - \widehat{\bsM}^\t - \widehat{\I_3}^\t} (\bsu \times \bsv)\\
&=& \Br{(\bsM + \I_3)^* - \bsM^* - \I_3^*} (\bsu \times \bsv)\\
&=& (\Tr{\bsM} \I_3 - \bsM) (\bsu \times \bsv) \\
&=& \Tr{\bsM} \bsu \times \bsv - \bsM (\bsu \times \bsv).
\end{eqnarray*}
We are done.
\end{proof}

Proposition \ref{p43} provides an identity that plays an important role in the subsequent reduction of candidate invariants. However, this identity was not explicitly given in the original work of King \emph{et al.}, making their reduction procedure difficult to follow. Specifically, it allows matrix actions on cross products to be rewritten in terms of simpler vector expressions, thereby facilitating the elimination of redundant candidate invariants in Subsection \ref{integrity-basis}.

The reduction identities above provide the tools for the graphical enumeration. As explained in Subsection~\ref{graphical_rules}, every polynomial invariant is a linear combination of complete contractions, and disconnected graphs represent products of the invariants generated by their connected components. It therefore suffices to consider connected graphs within the 31 cases listed in Table~\ref{cases}. The resulting candidate expressions are summarized in Table~\ref{candidate}.

The identities for products of two Levi-Civita tensors of the
same type, given in the proof below, reduce the number of
$\varepsilon$-nodes to at most one of each type.
Proposition~\ref{4c} and Corollary~\ref{CH} reduce long
$\bsC$-chains and higher matrix powers, while antisymmetry,
the determinant identity, and Proposition~\ref{p43} further
eliminate or reduce the remaining contractions.
The case-by-case analysis below details the enumeration of
the remaining configurations and reduces the additional
candidates to polynomial expressions in the known
21 generators.

The possible connection structures are constrained by the numbers
of legs of the building blocks. Vector nodes have one leg,
$C$-nodes have two, and only $\varepsilon$-nodes can branch.
After reducing pairs of $\varepsilon$-nodes of the same type
and excluding self-connections that vanish by antisymmetry,
a connected graph not involving the scalar $t$ therefore has
one of three forms. With no $\varepsilon$-node, it is a path
between two vectors or a closed $\bsC$-loop. With one
$\varepsilon$-node, its three legs lead to three vector endpoints.
With two $\varepsilon$-nodes of different types, one, two,
or three paths join them, and the remaining legs lead to
vector endpoints. The determinant identity reduces the
three-path case, while the other cases are examined below,
including the $Q_1,Q_2$ and $P_1,P_2$ configurations.
Together with the chain-reduction identities, this classification
reduces arbitrary connected contractions to the configurations
considered in the case-by-case analysis.

\begin{table}[H]
\centering
\caption{All the candidate generators.}
\small
\begin{minipage}{0.985\textwidth}
\noindent In the table, $\bsA$ and $\bsB$ denote arbitrary $A$-type and $B$-type vectors, respectively, i.e. vectors carrying a single uncontracted $\SO_A(3)$ or $\SO_B(3)$ index. Graphical examples are shown in Fig.~\ref{AB}.
\end{minipage}
\normalsize
\begin{tabular}{cc|cc}
\hline
notation & invariants & notation & invariants \\
\hline
$K_1$ & $t$ & & \\
\hline
$K_2$ & $\Inner{\bsC}{\bsC}$ & $U_1$ & $\Inner{(\bsa\cdot\cF^A)\bsC}{\bsC(\bsb\cdot\cF^B)}=2\innerm{\bsa}{\widehat \bsC}{\bsb}$ \\
\hline
$K_3$ & $\Inner{\bsa}{\bsa}$ & $U_2$ & $\innerm{\bsa}{\bsC\bsC^\t\bsC}{\bsb}$ \\
\hline
$K_4$ & $\Inner{\bsb}{\bsb}$ & $V_1$ & $\Inner{\bsC^\t\bsa\times\bsb}{\bsC^\t\bsC\bsb}$ \\
\hline
$K_5$ & $6\det(\bsC)$ & $V_2$ & $\Inner{\bsC^\t\bsa\times \bsb}{\bsC^\t\bsC\bsC^\t\bsa}$ \\
\hline
$K_6$ & $\innerm{\bsa}{\bsC}{\bsb}$ & $V_3$ & $\Inner{\bsC\bsC^\t\bsa\times \bsa}{\bsC\bsC^\t\bsC\bsb}$ \\
\hline
$K_7$ & $\Inner{\bsC\bsC^\t}{\bsC\bsC^\t}$ & $V_4$ & $\Inner{\bsC\bsC^\t\bsa\times \bsa}{\bsC\bsC^\t\bsC\bsC^\t\bsa}$ \\
\hline
$K_8$ & $\innerm{\bsa}{\bsC\bsC^\t}{\bsa}$ & $W_1$ & $\inner{\bsC\bsb\times \bsa}{\bsC\bsC^\t\bsa}$ \\
\hline
$K_9$ & $\innerm{\bsb}{\bsC^\t\bsC}{\bsb}$ & $W_2$ & $\Inner{\bsC\bsb\times \bsa}{\bsC\bsC^\t\bsC\bsb}$ \\
\hline
$X_1$ & $\innerm{\bsa}{\bsC\bsC^\t\bsC\bsC^\t}{\bsa}$ & $W_3$ & $\Inner{\bsC^\t\bsC\bsb\times \bsb}{\bsC^\t\bsC\bsC^\t\bsa}$ \\
\hline
$X_2$ & $\innerm{\bsb}{\bsC^\t\bsC\bsC^\t\bsC}{\bsb}$ & $W_4$ & $\Inner{\bsC^\t\bsC\bsb\times \bsb}{\bsC^\t\bsC\bsC^\t\bsC\bsb}$ \\
\hline
$P_1$ & $\Inner{(\bsa\cdot\cF^A)\bsC}{\bsC(\bsb\cdot\cF^B)\bsC^\t\bsC}$ & $Q_1$ & $\innerm{\bsA}{\bsC(\bsb\cdot\cF^B)\bsC^\t(\bsa\cdot\cF^A)\bsC}{\bsB}$ \\
\hline
$P_2$ & $\Inner{(\bsa\cdot\cF^A)\bsC\bsC^\t\bsC}{\bsC\bsC^\t\bsC(\bsb\cdot\cF^B)}$ & $Q_2$ & $\innerm{\bsA}{\bsC(\bsb\cdot\cF^B)\bsC^\t\bsC\bsC^\t(\bsa\cdot\cF^A)\bsC}{\bsB}$ \\
\hline
$Y_1$ & $\Inner{\bsa}{(\bsC\bsC^\t)^2\bsa\times\bsC\bsb} $ & $Z_1$ &$\Inner{\bsb}{(\bsC^\t\bsC)^2\bsb\times\bsC^\t\bsa} $ \\
\hline
$Y_2$ & $\Inner{\bsa}{(\bsC\bsC^\t)^2\bsa\times\bsC\bsC^\t\bsC\bsb}$ & $Z_2$ & $\Inner{\bsb}{(\bsC^\t\bsC)^2\bsb\times\bsC^\t\bsC\bsC^\t\bsa}$\\
\hline
\end{tabular}
\label{candidate}
\end{table}

\begin{thrm}[\cite{King2007}]\label{fundamental}
Let $V=\operatorname{vec}(\operatorname{End}(\bbC^2\otimes\bbC^2))$. Then, the ring
$\bbC[V]^{\sfG}$ of polynomials in the elements of $V$ that are
invariant under the local action of $\sfG =\GL_A(2)\times \GL_B(2)$
is generated by the \textbf{21} invariants:
\begin{eqnarray*}
\Set{K_1,\ldots,K_9}\cup\set{X_1,X_2}\cup\Set{U_1,U_2}\cup\Set{V_1,\ldots,V_4}\cup\Set{W_1,\ldots,W_4}.
\end{eqnarray*}
\end{thrm}

\begin{proof}
(i) For cases (1) and (2), it is evident that only invariants $\inner{\bsa}{\bsa}$ and $\inner{\bsb}{\bsb}$ can be constructed, see Figure \ref{all}.

\begin{figure}[htbp] \centering
\includegraphics[trim=5 9 7 10, clip, width=0.78\textwidth]{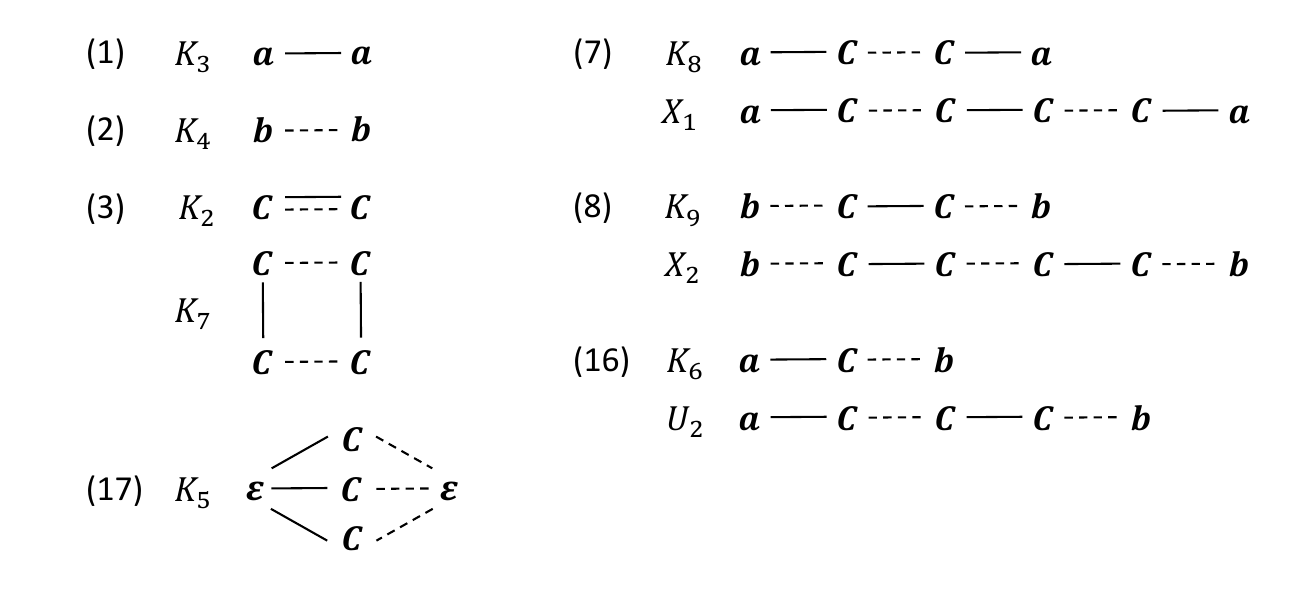}
\caption{Some candidate invariants.}
\label{all}
\end{figure}

(ii) All graphs containing two $\varepsilon$-nodes of the same type
may be eliminated by means of the identities
\begin{eqnarray*}
\varepsilon^A_{\alpha\beta\gamma} \varepsilon^A_{\alpha'\beta'\gamma'} &=& \delta_{\alpha\alpha'}\delta_{\beta\beta'}\delta_{\gamma\gamma'}+\delta_{\alpha\beta'}\delta_{\beta\gamma'}\delta_{\gamma\alpha'} + \delta_{\alpha\gamma'}\delta_{\beta\alpha'}\delta_{\gamma\beta'} -\delta_{\alpha\alpha'}\delta_{\beta\gamma'}\delta_{\gamma\beta'} - \delta_{\alpha\beta'}\delta_{\beta\alpha'}\delta_{\gamma\gamma'} - \delta_{\alpha\gamma'}\delta_{\beta\beta'}\delta_{\gamma\alpha'},\\
\varepsilon^B_{ijk}\varepsilon^B_{i'j'k'}&=&
\delta_{ii'}\delta_{jj'}\delta_{kk'}+\delta_{ij'}\delta_{jk'}\delta_{ki'} + \delta_{ik'}\delta_{ji'}\delta_{kj'}-\delta_{ii'}\delta_{jk'}\delta_{kj'} -
\delta_{ij'}\delta_{ji'}\delta_{kk'} -
\delta_{ik'}\delta_{jj'}\delta_{ki'}.
\end{eqnarray*}
Thus, when searching for generators, it suffices to consider contractions containing at most one \(\varepsilon\)-node of each type. From this, it is not possible to construct candidate
invariants for cases (4), (5).

And for cases (13), (14), to construct a connected graph that
achieves complete contraction of the indices, at least two
$\varepsilon$-nodes of the same type appear. For instance, the
connected graph composed of the fewest nodes is
$\varepsilon^B_{ijk}\varepsilon_{i'j'k}c^\t_{i\alpha}c_{\alpha
j}c^\t_{i'\beta}c_{\beta j'}$, see Figure \ref{14}. From the above
equation, we can obtain
\begin{eqnarray*}
\varepsilon^B_{ijk}\varepsilon^B_{i'j'k}c^\t_{i\alpha}c_{\alpha
j}c^\t_{i'\beta}c_{\beta
j'}&=&\sum_{i,j,k}\sum_{i',j'}\sum_{\alpha,
\beta}
\big(\delta_{ii'}\delta_{jj'}\delta_{kk}+\delta_{ij'}\delta_{jk}\delta_{ki'} + \delta_{ik}\delta_{ji'}\delta_{kj'}\\
&& -\delta_{ii'}\delta_{jk}\delta_{kj'} -
\delta_{ij'}\delta_{ji'} \delta_{kk}-
\delta_{ik}\delta_{jj'}\delta_{ki'}\big)c^\t_{i\alpha}c_{\alpha j}c^\t_{i'\beta}c_{\beta j'}\\
&=&3\sum_{i,j,\alpha,\beta} c^\t_{i\alpha}c_{\alpha j}c^\t_{i\beta}c_{\beta j}+\sum_{i,j,\alpha,\beta} c^\t_{i\alpha}c_{\alpha j}c^\t_{j\beta}c_{\beta i}+\sum_{i,j,\alpha,\beta} c^\t_{i\alpha}c_{\alpha j}c^\t_{j\beta}c_{\beta i}\\
&&-\sum_{i,j,\alpha,\beta} c^\t_{i\alpha}c_{\alpha j}c^\t_{i\beta}c_{\beta j}-3\sum_{i,j,\alpha,\beta} c^\t_{i\alpha}c_{\alpha j}c^\t_{j\beta}c_{\beta i}-\sum_{i,j,\alpha,\beta} c^\t_{i\alpha}c_{\alpha j}c^\t_{i\beta}c_{\beta j}\\
&=&5\tr{\bsC\bsC^\t\bsC\bsC^\t}-5\tr{\bsC\bsC^\t\bsC\bsC^\t}=0.
\end{eqnarray*}
Therefore, situations (13) and (14) can also be excluded.

\begin{figure}[htbp] \centering
\includegraphics[trim=5 14 7 10, clip, width=0.26\textwidth]{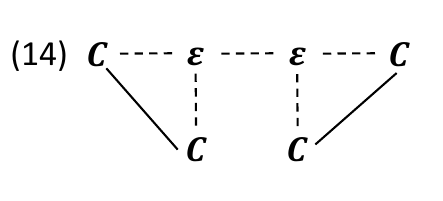}
\caption{In the case (14), the graph corresponding to
$\varepsilon^B_{ijk}\varepsilon^B_{i'j'k}c^\t_{i\alpha}c_{\alpha
j}c^\t_{i'\beta}c_{\beta j'}$.} \label{14}
\end{figure}

(iii) The antisymmetry of each $\varepsilon$-tensor implies that $\sum_{\alpha,\beta}\varepsilon^A_{\alpha\beta\gamma}a_\alpha a_\beta=0$ and $\sum_{i,j}\varepsilon^B_{ijk}b_i b_j=0$, for any vectors $\bsa$ and $\bsb$ having just a single uncontracted $A$- and $B$-type index, respectively.

Indeed, for any vector $\bsu\in\bbR^3$,
\begin{eqnarray*}
\sum^3_{\alpha,\beta=1}\varepsilon^A_{\alpha\beta\gamma}u_\alpha
u_\beta=\sum^3_{\alpha,\beta=1}\varepsilon^A_{\beta\alpha\gamma}u_\beta
u_\alpha
=\sum^3_{\alpha,\beta=1}(-\varepsilon^A_{\alpha\beta\gamma})
u_\alpha u_\beta =
-\sum^3_{\alpha,\beta=1}\varepsilon^A_{\alpha\beta\gamma}u_\alpha
u_\beta,
\end{eqnarray*}
where in the second equality we have relabeled the dummy indices.
Hence the sum must vanish. Equivalently, this is simply the
statement that $\bsu\times\bsu=0$ in terms of the cross product
defined by $[\bsu\times\bsv]_i=\sum_{j,k}\varepsilon_{ijk}u_jv_k$.
This is consistent with the graphical interpretation of the
$\varepsilon$-node as the cross product, which will be used
extensively in the reductions below.

This observation ensures that no two edges of any single
$\varepsilon$-node may be attached to the same vector. Consequently,
it is not possible to construct candidate invariants for cases (9)
and (12), see Figure~\ref{912}.

\begin{figure}[H] \centering\includegraphics[trim=5 8 7 10, clip,
width=1\textwidth]{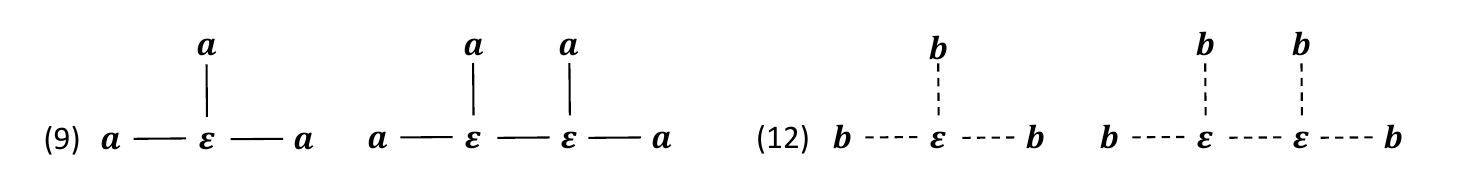} \caption{In both cases (9) and (12),
there must be a third-order tensor connecting two identical
vectors.} \label{912}
\end{figure}

(iv) We now examine the case where an $\varepsilon$-node has all
three of its edges connected to $\bsC$-nodes. The determinant of
$\bsC$ equals the triple product of its column vectors, hence it can
be expressed using the Levi-Civita symbol as
$$
\det(\bsC) = \sum_{\alpha, \beta, \gamma} \varepsilon_{\alpha \beta \gamma} c_{\alpha 1} c_{\beta 2} c_{\gamma 3}
= \varepsilon_{ijk} \sum_{\alpha, \beta, \gamma} \varepsilon_{\alpha \beta \gamma} c_{\alpha i} c_{\beta j} c_{\gamma k},
$$
where $ i, j, k$ are all distinct. Rearranging gives $\sum_{\alpha,
\beta, \gamma} \varepsilon_{\alpha \beta \gamma} c_{\alpha i}
c_{\beta j} c_{\gamma k} = \varepsilon_{ijk} \det(\bsC)$ for all
$i,j,k$.

For a general connected graph where all three edges of a single
$\varepsilon$-node are linked to $\bsC$-nodes, if the remaining
edges of these $\bsC$-nodes are contracted among themselves or with
other nodes, the above determinant expansion factors out a
$\det(\bsC)$ factor, leaving a lower-degree invariant. Consequently,
such graphs are reducible and do not yield new candidate invariants.
The only irreducible exception is the graph corresponding to
$$
\sum_{\alpha, \beta, \gamma, i, j, k = 1}^{3} \varepsilon^A_{\alpha \beta \gamma} c_{\alpha i} c_{\beta j} c_{\gamma k} \varepsilon^B_{ijk}=6 \det(\bsC),
$$
where the factor of $6$ arises from the sum over all permutations of
$i,j,k$. This graph is irreducible and is retained as the invariant
$K_5 = 6\det(\bsC)$.

For case (17), the graph contains only $\bsC$,
$\varepsilon^A_{\alpha\beta\gamma}$, and $\varepsilon^B_{ijk}$.
Since the same type of third-order tensor can only appear once, the
only connected graph that can constitute a candidate invariant is
$6\det(\bsC)$, see Figure~\ref{all}. All other connected graphs can
be decomposed.

In case (22), the graph contains $\bsb$, $\bsC$, and
$\varepsilon^A_{\alpha\beta\gamma}$. Since
$\varepsilon^A_{\alpha\beta\gamma}$ can appear only once, its three
edges must all be connected to $\bsC$-nodes. However, the remaining
edges of these $\bsC$-nodes can only connect to $\bsb$ or among
themselves. A direct inspection of all possible connection patterns
shows that no connected graph can be formed that is not already
reducible to products of lower-degree invariants. Case (23) is the
dual, with $\varepsilon^B_{ijk}$ replacing
$\varepsilon^A_{\alpha\beta\gamma}$ and $\bsa$ replacing $\bsb$ .

Similarly, case (29) contains $\bsa$, $\bsC$,
$\varepsilon^A_{\alpha\beta\gamma}$, and $\varepsilon^B_{ijk}$,
where both $\varepsilon^A_{\alpha\beta\gamma}$, and
$\varepsilon^B_{ijk}$ can appear only once. Any graph would require
the three edges of $\varepsilon^B$-node to be attached to
$\bsC$-nodes, leading to graphs that either vanish or factor through
lower-degree invariants. Thus no new candidate invariants arise from
this case. Case (30) is the dual.

Thus, apart from the determinant invariant $K_5$, a graph
representing a new generator may be assumed to have at least
one vector node directly attached to each $\varepsilon$-node.
This explains the initial configurations used in the
enumeration diagrams for cases (24), (26), and (31).

(v) By Proposition \ref{4c}, the graph of any candidate invariant may contain a
string of no more than 4 $\bsC$ nodes.

Therefore, in situation (3), only two possible candidate invariants,
$\inner{\bsC}{\bsC}$ and $\inner{\bsC\bsC^\t}{\bsC\bsC^\t}$, can be
constructed. Similarly, in case (7), only candidate invariants
$\innerm{\bsa}{\bsC\bsC^\t\bsC\bsC^\t}{\bsa}$ and
$\innerm{\bsa}{\bsC\bsC^\t}{\bsa}$ can be constructed, while in case
(8), only candidate invariants
$\innerm{\bsb}{\bsC^\t\bsC\bsC^\t\bsC}{\bsb}$ and
$\innerm{\bsb}{\bsC^\t\bsC}{\bsb}$ can be constructed. Moreover,
situation (16) can only generate candidate invariants
$\innerm{\bsa}{\bsC}{\bsb}$ and
$\innerm{\bsa}{\bsC\bsC^\t\bsC}{\bsb}$. The above invariants are
shown in Figure \ref{all}.

From the 26 cases discussed above, we obtain 11 candidate invariants
in total. Among the remaining five cases, namely (24), (25), (26),
(27) and (31), cases (24) and (25) are dual to each other, and cases
(26) and (27) are also dual to each other. Hence it suffices to
consider only (24), (26) and (31).

(vi) The antisymmetry of the $\varepsilon$-tensors also guarantees
that no two edges of any $\varepsilon$-node may be linked to a loop
formed solely by $\bsC$-nodes. Indeed, any closed chain constructed
from an even number of $\bsC$ and $\bsC^\t$ factors, such as
$(\bsC\bsC^\t)^n$ or $(\bsC^\t \bsC)^n$, is symmetric. Consequently,
its contraction with an $\varepsilon$-tensor vanishes identically.

To see this explicitly, consider $(\bsC\bsC^\t)^n$, for which
$[(\bsC\bsC^\t)^n]_{\alpha\beta}=[(\bsC\bsC^\t)^n]_{\beta\alpha}$.
Thus
\begin{eqnarray*}
\sum_{\alpha,\beta}\varepsilon^A_{\alpha\beta\gamma}[(\bsC\bsC^\t)^n]_{\alpha\beta}
&=&\sum^3_{\alpha,\beta=1}\varepsilon^A_{\beta\alpha\gamma}[(\bsC\bsC^\t)^n]_{\beta\alpha}\\
&=&\sum^3_{\alpha,\beta=1}[-\varepsilon^A_{\alpha\beta\gamma}][(\bsC\bsC^\t)^n]_{\alpha\beta}\\
&=&-\sum^3_{\alpha,\beta=1}\varepsilon^A_{\alpha\beta\gamma}[(\bsC\bsC^\t)^n]_{\alpha\beta}
\end{eqnarray*}
where we have relabeled the dummy indices in the second step. Hence
the sum must vanish. Analogously, we get that
$\sum_{i,j}\varepsilon^B_{ijk}[(\bsC^\t\bsC)^n]_{ij}=0$.

(vii) We now consider case (24), where connected graphs are
constructed using only $\bsb, \bsC$, and $\varepsilon^B_{ijk}$.
Since at most one $\varepsilon$-node can appear in a connected
graph, the enumeration starts from this node, as shown in
Figure~\ref{24}. From the figure, only the connected graphs marked
in red yield a candidate invariant, namely
$W_4=\inner{\bsC^\t\bsC\bsb\times\bsb}{\bsC^\t\bsC\bsC^\t\bsC\bsb}$.
This is because the invariants corresponding to the other graphs can
either be generated by other invariants or are identically zero.
Dually, in case (25), one obtains the candidate invariant
$V_4=\inner{\bsC\bsC^\t\bsa\times\bsa}{\bsC\bsC^\t\bsC\bsC^\t\bsa}$.

(viii) We next address case (26), which constructs the connected
graph using only $\bsa, \bsb, \bsC$, and
$\varepsilon^A_{\alpha\beta\gamma}$. The connected graphs are
constructed starting from $\varepsilon^A_{\alpha\beta\gamma}$, see
Figure \ref{26} in Appendix \ref{graph}. By enumerating all the
connected graphs, we obtain 5 candidate
invariants for case (26), namely $W_1$, $W_2$, $V_3$, $Y_1$ and
$Y_2$. Dually, in case (27), the candidate invariants are $V_1$,
$V_2$, $W_3$, $Z_1$ and $Z_2$.

From Propositions~\ref{4c} and \ref{p43}, it holds that
$\widehat{\bsC\bsC^\t} =
(\bsC\bsC^\t)^2-\Tr{\bsC\bsC^\t}\bsC\bsC^\t+\Tr{\widehat{\bsC\bsC^\t}}\I_3$
and $\bsC\bsu\times\bsC\bsv=\widehat{\bsC}(\bsu\times\bsv)$ for any
$\bsu,\bsv,\bsC$. Then we obtain
\begin{eqnarray*}
Y_1&=& \Inner{\bsa}{(\bsC\bsC^\t)^2\bsa\times\bsC\bsb} \\
&=& \Inner{\bsa}{\widehat{\bsC\bsC^\t}\bsa\times\bsC\bsb} +
\Tr{\bsC\bsC^\t}\Inner{\bsa}{\bsC\bsC^\t\bsa\times\bsC\bsb}-\Tr{\widehat{\bsC\bsC^\t}}\Inner{\bsa}{\bsa\times\bsC\bsb}\\
&=& \Inner{\widehat{\bsC\bsC^\t}\bsa}{\bsC\bsb\times\bsa} +
\Tr{\bsC\bsC^\t}\Inner{\bsa}{\bsC\bsC^\t\bsa\times\bsC\bsb}\\
&=& \Inner{\bsa}{\bsC\bsC^\t\bsC\bsb\times\bsC\bsC^\t\bsa} + \Tr{\bsC\bsC^\t}\Inner{\bsa}{\bsC\bsC^\t\bsa\times\bsC\bsb}\\
&=&\Inner{\bsC\bsC^\t\bsa\times \bsa}{\bsC\bsC^\t\bsC\bsb} + \Tr{\bsC\bsC^\t}\Inner{\bsC\bsb\times\bsa}{\bsC\bsC^\t\bsa}\\
&=&V_3+K_2W_1
\end{eqnarray*}
and
\begin{eqnarray*}
Y_2&=&\Inner{\bsa}{(\bsC\bsC^\t)^2\bsa\times\bsC\bsC^\t\bsC\bsb}\\
&=& \Inner{\bsa}{\widehat{\bsC\bsC^\t}\bsa\times\bsC\bsC^\t\bsC\bsb}
+ \Tr{\bsC\bsC^\t}\Inner{\bsa}{\bsC\bsC^\t\bsa\times
\bsC\bsC^\t\bsC\bsb}-\Tr{\widehat{\bsC\bsC^\t}}\Inner{\bsa}{\bsa\times\bsC\bsC^\t\bsC\bsb}\\
&=&\Inner{\bsC\bsC^\t\bsC\bsb\times\bsa}{\widehat{\bsC\bsC^\t}\bsa} + \Tr{\bsC\bsC^\t}\Inner{\bsa}{\bsC\bsC^\t\bsa\times \bsC\bsC^\t\bsC\bsb}\\
&=& \Inner{\bsC\bsC^\t\bsC\bsb\times\bsa}{\widehat{\bsC\bsC^\t}\bsa} - \Tr{\bsC\bsC^\t}\Inner{\bsC\bsC^\t\bsa\times\bsa}{\bsC\bsC^\t\bsC\bsb}\\
&=& \Inner{\bsC\bsC^\t\bsC\bsC^\t\bsC\bsb\times\bsC\bsC^\t\bsa}{\bsa} - \Tr{\bsC\bsC^\t}\Inner{\bsC\bsC^\t\bsa\times\bsa}{\bsC\bsC^\t\bsC\bsb}\\
&=&\Inner{\bsC\bsC^\t\bsa\times\bsa}{\bsC\bsC^\t\bsC\bsC^\t\bsC\bsb}
-\Tr{\bsC\bsC^\t}\Inner{\bsC\bsC^\t\bsa\times\bsa}{\bsC\bsC^\t\bsC\bsb}.
\end{eqnarray*}
Using $\bsC\bsC^\t\bsC\bsC^\t\bsC=
K_2\bsC\bsC^\t\bsC+\det(\bsC)\widehat\bsC - \frac12(K_2^2-K_7)\bsC$,
we obtain
\begin{eqnarray*}
Y_2 &=&K_2\Inner{\bsC\bsC^\t\bsa\times\bsa}{\bsC\bsC^\t\bsC\bsb} +
\det(\bsC)\Inner{\bsC\bsC^\t\bsa\times\bsa}{\widehat\bsC\bsb}-\tfrac12(K_2^2-K_7)\Inner{\bsC\bsC^\t\bsa\times\bsa}{\bsC\bsb}\\
&&-K_2\Inner{\bsC\bsC^\t\bsa\times\bsa}{\bsC\bsC^\t\bsC\bsb}\\
&=&\det(\bsC)\Inner{\bsC\bsC^\t\bsa\times\bsa}{\widehat\bsC\bsb}-\tfrac12(K_2^2-K_7)\Inner{\bsC\bsC^\t\bsa\times\bsa}{\bsC\bsb}\\
&=&\det(\bsC)\Inner{\bsC^\t\bsC\bsC^\t\bsa\times\bsC^\t\bsa}{\bsb}-\tfrac12(K_2^2-K_7)\Inner{\bsC\bsC^\t\bsa\times\bsa}{\bsC\bsb}\\
&=&\det(\bsC)\Inner{\bsC^\t\bsa\times\bsb}{\bsC^\t\bsC\bsC^\t\bsa}+\tfrac12(K_2^2-K_7)\Inner{\bsC\bsb\times\bsa}{\bsC\bsC^\t\bsa}\\
&=&\frac16K_5V_2+\frac12(K_2^2-K_7)W_1.
\end{eqnarray*}

Analogously, in case (27), it holds that
\begin{eqnarray*}
Z_1&=&\Inner{\bsb}{(\bsC^\t\bsC)^2\bsb\times\bsC^\t\bsa}\\
&=&\Inner{\bsC^\t\bsC\bsb\times \bsb}{\bsC^\t\bsC\bsC^\t\bsa} + \Tr{\bsC\bsC^\t}\Inner{\bsC^\t\bsa\times\bsb}{\bsC^\t\bsC\bsb}\\
&=&W_3+K_2V_1
\end{eqnarray*} and
\begin{eqnarray*}
Z_2&=&\Inner{\bsb}{(\bsC^\t\bsC)^2\bsb\times\bsC^\t\bsC\bsC^\t\bsa}\\
&=&\det(\bsC)\Inner{\bsC\bsb\times\bsa}{\bsC\bsC^\t\bsC\bsb}+\tfrac12(K_2^2-K_7)\Inner{\bsC^\t\bsa\times\bsb}{\bsC^\t\bsC\bsb}\\
&=&\frac16K_5W_2+\frac12(K_2^2-K_7)V_1.
\end{eqnarray*}

Therefore, the candidate invariants $Y_1$, $Y_2$, $Z_1$, and $Z_2$
are reducible and hence need not be included in the generating set.

(ix) For the last case, all building blocks are needed to construct
a connected graph (see Figures~\ref{31a}--\ref{31c} in
Appendix~\ref{graph}). By the determinant reduction above,
we only need to consider graphs in which each $\varepsilon$-node
is directly attached to a vector. The other two legs of each
$\varepsilon$-node can connect to the other $\varepsilon$-node
or to vectors through $\bsC$-nodes. Since the graph is connected,
at least one path must join the two $\varepsilon$-nodes.
Such a path contains an odd number of $\bsC$-nodes, and
Proposition~\ref{4c} reduces paths with five or more nodes.
Thus, each joining path need contain only one or three
$\bsC$-nodes.

If only one path joins the two $\varepsilon$-nodes, the remaining
legs connect to vectors. This gives the two forms $Q_1$ and $Q_2$
shown in Figure~\ref{Q}, where $\bsA$ and $\bsB$ denote arbitrary
$A$-type and $B$-type vectors, respectively.
If two paths join the nodes, both may contain one $\bsC$-node,
one may contain one and the other three, or both may contain
three. These give $U_1$, $P_1$, and $P_2$, respectively.
Figures~\ref{31a}--\ref{31c} show these possibilities,
whose algebraic reductions are given below.

\begin{figure}[H] \centering
\includegraphics[trim=5 14 5 10, clip, width=0.6\textwidth]{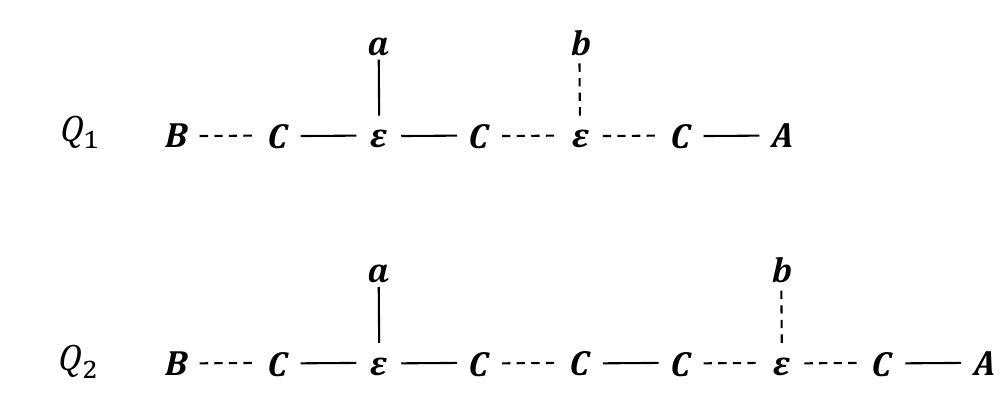}
\caption{The two types of invariants in situation (31).} \label{Q}
\end{figure}

Using \cite[Lemmas B1 and B2]{Zhang2025}, it holds that
$\bsM(\bsu\cdot\cF)\bsM^\t=(\widehat{\bsM}\bsu\cdot\cF)$ and
$(\bsu\cdot\cF)^\t(\bsv\cdot\cF)=\inner{\bsv}{\bsu}\I_3-\out{\bsv}{\bsu}$
for any $\bsu,\bsv,\bsM$. Using the two identities, we obtain
\begin{eqnarray*}
Q_1&=&\innerm{\bsA}{\bsC(\bsb\cdot\cF^B)\bsC^\t(\bsa\cdot\cF^A)\bsC}{\bsB}=\innerm{\bsA}{(\widehat\bsC\bsb\cdot\cF^A)(\bsa\cdot\cF^A)\bsC}{\bsB}\\
&=& \innerm{\bsA}{\Pa{\out{\bsa}{\widehat\bsC\bsb} - \Inner{\bsa}{\widehat\bsC\bsb}\I_3}\bsC}{\bsB}\\
&=&\det(\bsC)\Inner{\bsA}{\bsa}\Inner{\bsb}{\bsB} -
\innerm{\bsa}{\widehat \bsC}{\bsb}\innerm{\bsA}{\bsC}{\bsB}.
\end{eqnarray*}
and
\begin{eqnarray*}
Q_2&=&\innerm{\bsA}{\bsC(\bsb\cdot\cF)\bsC^\t\bsC\bsC^\t(\bsa\cdot\cF)\bsC}{\bsB}=-\innerm{\bsA}{(\widehat\bsC\bsb\cdot\cF)\bsC\bsC^\t}{\bsa\times\bsC
\bsB} \\
&=&\innerm{\widehat\bsC\bsb\times\bsA}{\bsC\bsC^\t}{\bsC
\bsB\times\bsa}.
\end{eqnarray*}
By Proposition \ref{p43}, we have
\begin{eqnarray*}
(\bsC\bsC^\t)(\bsC \bsB\times\bsa)&=& \Inner{\bsC}{\bsC}(\bsC
\bsB\times\bsa) - (\bsC\bsC^\t)^\t\bsC \bsB\times\bsa - \bsC
\bsB\times(\bsC\bsC^\t)^\t\bsa\\
&=&\Inner{\bsC}{\bsC}(\bsC \bsB\times\bsa) - \bsC\bsC^\t\bsC
\bsB\times\bsa - \bsC \bsB\times\bsC\bsC^\t\bsa.
\end{eqnarray*}
Thus
\begin{eqnarray*}
Q_2&=&\Inner{\bsC}{\bsC}\Inner{\widehat\bsC\bsb\times\bsA}{\bsC\bsB\times\bsa}
-\Inner{\widehat\bsC\bsb\times\bsA}{\bsC\bsC^\t\bsC \bsB\times\bsa}-
\Inner{\widehat\bsC\bsb\times\bsA}{\bsC \bsB\times\bsC\bsC^\t\bsa}.
\end{eqnarray*}
Using the following identity:
$\Inner{\bsu\times\bsv}{\bsx\times\bsy}=\Inner{\bsu}{\bsx}\Inner{\bsv}{\bsy}
- \Inner{\bsu}{\bsy}\Inner{\bsv}{\bsx}$ for any
$\bsu,\bsv,\bsx,\bsy\in\bbR^3$, we obtain
\begin{eqnarray*}
\Inner{\widehat{\bsC}\bsb\times\bsA}{\bsC\bsB\times\bsa}
&=&\det(\bsC)\Inner{\bsA}{\bsa}\Inner{\bsb}{\bsB}
-\innerm{\bsa}{\widehat{\bsC}}{\bsb}\innerm{\bsA}{\bsC}{\bsB},\\
\Inner{\widehat{\bsC}\bsb\times\bsA}{\bsC\bsC^\t\bsC\bsB\times\bsa}
&=&\det(\bsC)\Inner{\bsA}{\bsa}\innerm{\bsb}{\bsC^\t\bsC}{\bsB}
-\innerm{\bsa}{\widehat{\bsC}}{\bsb}\innerm{\bsA}{\bsC\bsC^\t\bsC}{\bsB},\\
\Inner{\widehat{\bsC}\bsb\times\bsA}{\bsC\bsB\times\bsC\bsC^\t\bsa}
&=&\det(\bsC)\Bigl[\innerm{\bsA}{\bsC\bsC^\t}{\bsa}\Inner{\bsb}{\bsB}
-\innerm{\bsa}{\bsC}{\bsb}\innerm{\bsA}{\bsC}{\bsB}\Bigr].
\end{eqnarray*}
Therefore we get that
\begin{eqnarray*}
Q_2&=&\innerm{\bsa}{\widehat\bsC}{\bsb}\Br{\innerm{\bsA}{\bsC\bsC^\t\bsC}{\bsB}
-\innerm{\bsA}{\bsC}{\bsB}\Inner{\bsC}{\bsC}} \\
&&+\det(\bsC)\Big[\Inner{\bsA}{\bsa}\Inner{\bsb}{\bsB}\Inner{\bsC}{\bsC}
+\innerm{\bsa}{\bsC}{\bsb}\innerm{\bsA}{\bsC}{\bsB} \\
&&- \innerm{\bsA}{\bsC\bsC^\t}{\bsa}\Inner{\bsb}{\bsB} -
\Inner{\bsA}{\bsa}\innerm{\bsb}{\bsC^\t\bsC}{\bsB}\Big].
\end{eqnarray*}

Furthermore, the candidates $P_1$ and $P_2$ are not displayed separately in \cite{King2007}, but are covered by the graphical identities in Figure~5 of that reference, together with the
stated replacement of vector products by arbitrary matrices.
For completeness, we explicitly derive their reductions below.
For $P_1$, we have
\begin{eqnarray*}
P_1&=&\inner{(\bsa\cdot\cF)\bsC}{\bsC(\bsb\cdot\cF)\bsC^\t\bsC}=\tr{\bsC^\t(\bsa\cdot\cF)^\t\bsC(\bsb\cdot\cF)\bsC^\t\bsC}\\
&=&-\tr{(\widehat{\bsC^\t}\bsa\cdot\cF)(\bsb\cdot\cF)\bsC^\t\bsC}=-\tr{(\out{\bsb}{\widehat{\bsC^\t}\bsa}-\inner{\widehat{\bsC^\t}\bsa}{\bsb}\I_3)\bsC^\t\bsC} \\
&=&-\det(\bsC)\innerm{\bsa}{\bsC}{\bsb}+\innerm{\bsa}{\widehat{\bsC}}{\bsb}\inner{\bsC}{\bsC}.
\end{eqnarray*}
Using
$\widehat{\bsC^\t}\widehat{\bsC}=(\bsC^\t\bsC)^2-\inner{\bsC}{\bsC}\bsC^\t\bsC+\inner{\widehat{\bsC}}{\widehat{\bsC}}\I_3$
and
$\inner{\widehat{\bsC}}{\widehat{\bsC}}=\frac12\Pa{\inner{\bsC}{\bsC}^2-\inner{\bsC^\t\bsC}{\bsC^\t\bsC}}$,
we obtain
\begin{eqnarray*}
P_2&=&\inner{(\bsa\cdot\cF)\bsC\bsC^\t\bsC}{\bsC\bsC^\t\bsC(\bsb\cdot\cF)}=\Tr{\bsC^\t\bsC\bsC^\t(\bsa\cdot\cF)^\t\bsC\bsC^\t\bsC(\bsb\cdot\cF)}\\
&=&-\Tr{(\widehat{\bsC^\t\bsC\bsC^\t}\bsa\cdot\cF)(\bsb\cdot\cF)}=-\Tr{(\out{\bsb}{\widehat{\bsC^\t\bsC\bsC^\t}\bsa}-\inner{\widehat{\bsC^\t\bsC\bsC^\t}\bsa}{\bsb}\I_3)}\\
&=&2\innerm{\bsa}{\widehat{\bsC}\widehat{\bsC^\t}\widehat{\bsC}}{\bsb}=2\Pa{\det(\bsC)\innerm{\bsa}{\bsC\bsC^\t\bsC}{\bsb}-\det(\bsC)\inner{\bsC}{\bsC}\innerm{\bsa}{\bsC}{\bsb}+\inner{\widehat{\bsC}}{\widehat{\bsC}}\innerm{\bsa}{\widehat{\bsC}}{\bsb}}\\
&=&2\det{(\bsC)}[\innerm{\bsa}{\bsC\bsC^\t\bsC}{\bsb}-\inner{\bsC}{\bsC}\innerm{\bsa}{\bsC}{\bsb}]+\innerm{\bsa}{\widehat{\bsC}}{\bsb}[\inner{\bsC}{\bsC}^2-\inner{\bsC^\t\bsC}{\bsC^\t\bsC}].
\end{eqnarray*}
As mentioned above, the candidate invariants $Q_1,Q_2, P_1$, and
$P_2$ need not be included in the generating set.

In summary, by enumerating all connected undirected graphs, we
obtain the candidate invariants, which are presented in Table
\ref{candidate}. Subsequently, through explicit computation, some of
these invariants can be ruled out, leading to the result stated in Theorem \ref{fundamental}.
\end{proof}

\section{Concluding remarks}\label{s5}

In this paper, we have systematically revisited the work of King
\emph{et al.} on the ring of local invariants for mixed two-qubit
systems. Starting from the physical motivation, we recalled the
basic notion of local unitary invariants. By employing the Bloch expansion and the vectorization framework, we expressed the local conjugation action of \(\G=\GL(2)\times\GL(2)\) as a linear representation on the full operator space \(V\cong\bbC^{16}\). On this basis, we presented a detailed computation of the Molien series, providing a complete derivation from the contour integral in
Molien's theorem to its closed-form expression. We then used the graphical tensor method of King \emph{et al.}
to systematically reconstruct the known generating set of 21
invariants by enumerating connected contractions of the basic
building blocks. In particular, we explicitly enumerated and
algebraically reduced the candidates $P_1$ and $P_2$, which are
covered by their general graphical identities but were not
separately displayed there. Our contribution lies in supplying
detailed intermediate calculations and making the graphical
reconstruction explicit.

A natural next step is to generalize this approach to qudit-qudit
systems. However, unlike the well-developed theory in the two-qubit
case, the invariant rings for qubit-qutrit systems remain in a
preliminary stage. The two-qubit case benefits from the surjective
Lie group homomorphism $\SU(2)\to\SO(3)$, which effectively reduces
the $\SU(2)\times\SU(2)$ action on Bloch parameters to orthogonal
transformations in $\SO(3)\times\SO(3)$. This special property is
precisely what makes the graphical method of \cite{King2007} so
efficient. For qubit-qutrit systems, the local unitary group is
$\SU(2)\times\SU(3)$; however, the map $\SU(3)\to\SO(8)$ is not
surjective---and, more generally, no such surjective Lie group
homomorphism $\SU(d)\to\SO(d^2-1)$ exists for any integer $d>2$.
Consequently, the specific graphical construction and reduction rules used in the two-qubit case cannot be directly ported to the qubit-qutrit setting.

Gerdt \emph{et al.} sought to characterize the polynomial invariants
of qubit-qutrit mixed states using the Casimir invariants of
$\SU(6)$ and computed the corresponding Molien series
\cite{Gerdt2011}. They explicitly highlighted the substantial
computational difficulties inherent in constructing local unitary
polynomial invariants. While these works lay a valuable
group-theoretic foundation, their results remain incomplete and do
not yet yield a comprehensive description on par with the two-qubit
case.

In our previous work \cite{Zhang2025}, we established a formal
connection between Makhlin's fundamental invariants and \emph{local
unitary Bargmann invariants}\footnote{Note: All local unitary
Bargmann invariants are measurable quantities by quantum circuits.},
and successfully employed the latter to completely characterize
local unitary equivalence for two-qubit systems. Building on this,
we conjectured in \cite{Zhang2026} that the local unitary orbit of
an arbitrary multipartite state is completely determined by its
Bargmann invariants. However, Ma and Shi \cite{MaShi2025} recently
disproved this conjecture for the equal-dimension case $\bbC^d\ot
\bbC^d$ with $d\geqslant 3$, by constructing explicit
counterexamples. This demonstrates that the Bargmann-invariant
approach is not universally valid for equal dimensions $d\geqslant
3$. In subsequent work \cite{ZhangXie2026Insufficient}, we constructed qubit-qutrit states that agree on all ordinary marginal-word Bargmann invariants but are not LU equivalent. Taken together, these results indicate the limitations	of ordinary marginal-word Bargmann invariants beyond the two-qubit setting and highlight the special role of two-qubit systems.

In summary, a complete characterization of the ring of local
invariants for qubit-qutrit systems remains a challenging yet
worthwhile direction for future research. It represents not only a
natural extension of two-qubit theory, but also a potential stepping
stone toward understanding the invariant structure of general $d_1
\ot d_2$ systems.

\subsection*{Data availability statement}

No new data were created or analyzed in this study.

\newpage
\appendix
\appendixpage
\addappheadtotoc

\section{Derivation of Eq.~\eqref{MJK}}\label{a}

Fix a secondary invariant $J_k$ and consider the set $\cR_k = J_k
\cdot \bbC[K_1, \dots, K_n]$. By the direct sum decomposition
\eqref{dsum}, the subspaces $\cR_k$ are linearly independent as
subspaces, hence
$$
n_m = \sum_{k=0}^{r} n_{m,k},
$$
where $n_{m,k}$ denotes the number of degree-$m$ homogeneous
invariants in $\cR_k$.

Let $P \in \bbC[K_1, \dots, K_n]$, then $\deg(J_k P) = \deg J_k +
\deg P$. If $m - \deg J_k \geqslant 0$, then $n_{m,k}$ equals the
number of degree $d = m - \deg J_k$ homogeneous polynomials in
$\bbC[K_1, \dots, K_n]$, denoted by $n_d$; otherwise $n_{m,k} = n_d
= 0$. Since $K_1, \dots, K_n$ are algebraically independent, a basis
for the degree-$d$ homogeneous polynomials is given by $K_1^{a_1}
\cdots K_n^{a_n}$, where $(a_1, \dots, a_n)$ satisfy $\sum_{i=1}^n
a_i \deg K_i = d$.

So $n_d$ is the number of nonnegative integer solutions $(a_1,
\dots, a_n)$ to this equation. For each variable $ K_i $, its
exponent $ a_i $ contributes a factor $ q^{a_i \deg K_i} $ to
the generating function. Since the choices of $ a_1, \dots, a_n $
are independent, the generating function for all such solutions is
the product of the individual geometric series:
$$
\sum_{d=0}^{\infty} n_d \, q^d = \prod_{i=1}^{n} \Pa{
\sum_{a_i=0}^{\infty} q^{a_i \deg K_i}}= \prod_{i=1}^{n} \frac{1}{1
- q^{\deg K_i}}.
$$
Here the last equality is understood in the sense of formal power
series. If one interprets it as an ordinary numerical series, the
identity requires $\abs{q} < 1 $ for convergence. Therefore,
\begin{eqnarray*}
 M_{\G}(q) &=& \sum_{m=0}^{\infty} n_m q^m=\sum_{m=0}^{\infty}\sum_{k=0}^{r} n_{m,k}q^m\\
&=& \sum_{k=0}^{r} \sum_{m=\deg{J_k}}^{\infty} n_{m-\deg{J_k}}q^m= \sum_{k=0}^{r} \sum_{d=0}^{\infty}n_d q^dq^{\deg J_k} \\
&=& \sum_{k=0}^{r} q^{\deg J_k} \prod_{i=1}^{n} \frac{1}{1 - q^{\deg
K_i}}=\frac{\sum_{k=0}^{r} q^{\deg J_k}}{\prod_{i=1}^{n} (1 -
q^{\deg K_i})}.
\end{eqnarray*}

Here, the second equality is given by the direct decomposition $	\bbC[V]^G = \bigoplus_{k=0}^r J_k \cdot \bbC[K_1,\dots,K_n] $. The third equality follows from interchanging the order of summation and using $n_{m,k}=0$ for $m<\deg J_k$. The fourth equality follows from $d=m-\deg J_k$, $n_{m,k}=n_d$, and $q^m=q^d q^{\deg J_k}$. The fifth equality is given by the generating function for \(n_d\).

\section{Molien's theorem and its proof}\label{Molien}
In this appendix we provide a self-contained proof of Molien's theorem for compact Lie groups, which is the central tool used in Section \ref{s3} to compute the Molien series for the two-qubit invariant ring.

\begin{thrm}[Molien's theorem for a compact Lie group \cite{Derksen2015}]
Let $\sfG$ be a compact Lie group equipped with a normalized Haar
measure $\dif\mu$ satisfying $\int_{\sfG} \dif\mu(g) = 1$. Given a
rational representation $\bsT : \sfG \to \GL(V)$ with $\bsT(g) =
\bsT_g \in \GL(V)$ for each $g \in \sfG$, then the Molien series
takes the integral form
$$
M_{\sfG}(q) = \int_{\sfG} \frac{\dif\mu(g)}{\det(\I_V - q\bsT_g)}.
$$
\end{thrm}

\begin{proof}
Let $\bsT: \G \to \GL(V)$ be a continuous representation on a
finite-dimensional complex vector space $V$, with $\dim_{\bbC} V =
n$. Denote by $\bbC[V]_m$ the space of homogeneous polynomials of
degree $m$ on $V$, which is naturally identified with the $m$-th
symmetric power $\operatorname{Sym}^m(V^*)$. The group $\G$ acts on
any polynomial $f\in\bbC[V]$ by
$$
(g \cdot f)(\bsv) := f(\bsT_{g^{-1}} \bsv), \quad
\bsv \in V.
$$ This defines a left action: $(gh)\cdot f=g\cdot
(h\cdot f)$ for any $g,h\in \G$. For each $m \geqslant 0$, define
the \emph{Reynolds operator} (averaging projector)
$$
\widetilde{\cT}_m : \bbC[V]_m \to \bbC[V]_m^\G, \quad
\widetilde{\cT}_m(f) := \int_\G (g \cdot f) \dif \mu(g).
$$
Since $\G$ is compact and the representation is continuous, the
integral converges in the finite-dimensional space $\bbC[V]_m$.
Moreover, $\widetilde{\cT}_m$ is a linear projection onto the
invariant subspace $\bbC[V]_m^\G$.
\begin{itemize}
\item For any $f \in \bbC[V]_m$ and any $h\in\G$, it holds that
\begin{eqnarray*}
h\cdot \widetilde{\cT}_m(f) = \int_\G h\cdot(g\cdot f)\,\dif\mu(g) =
\int_\G (hg)\cdot f\,\dif\mu(g) = \int_\G g'\cdot f\,\dif\mu(g') =
\widetilde{\cT}_m(f),
\end{eqnarray*}
where the third equality follows from the left-invariance of the
Haar measure via the substitution $g'=hg$. Thus the image lies in
the invariant subspace, $ \widetilde{\cT}_m(f) \in \bbC[V]_m^\G$.
\item If $f\in\bbC[V]_m^\G$, then
$g\cdot f=f$ for all $g\in\G$, so $\widetilde{\cT}_m(f) = \int_\G
(g\cdot f) \,\dif\mu(g) = \int_\G f \dif\mu(g) = f$. Thus
$f\in\im(\widetilde{\cT}_m)$ for any $f\in\bbC[V]^\G_m$.
\end{itemize}
These two properties imply that for any $f$, $\widetilde{\cT}_m(f)$
is invariant, and $\widetilde{\cT}_m$ acts as the identity on
invariant polynomials. Thus $\widetilde{\cT}_m$ is a projection onto
$\bbC[V]_m^\G$. Since the trace of any projection equals the
dimension of its image, then
$$
\dim \bbC[V]_m^\G = \Tr{\widetilde{\cT}_m}.
$$
Since the trace is linear and continuous, we may interchange it with
the integral (Fubini's theorem for compact groups):
$$
\Tr{\widetilde{\cT}_m} = \int_\G \Tr{\cT_g^m} \dif\mu(g),
$$
where $\cT_g^m : \bbC[V]_m \to \bbC[V]_m$ is defined by $\cT_g^m(f)
:= g\cdot f$, and $\Tr{\cT_g^m}$ denotes the trace of this linear
map on the finite-dimensional space $\bbC[V]_m$.

Now fix $g \in \G$. Let $\lambda_1, \dots, \lambda_n$ be the
eigenvalues of $\bsT_g$ on $V$. Since $\G$ is compact, $\bsT_g$ is
diagonalizable and all eigenvalues lie on the unit circle. Choose a
basis $\set{\bsv_1,\dots,\bsv_n}$ of eigenvectors of $\bsT_g$, with
$\bsT_g \bsv_i = \lambda_i \bsv_i$, and let $\set{x_1,\dots,x_n}$ be
the dual basis of $V^*$, so that $x_i(\bsv_j) = \delta_{ij}$. Then a
basis of $\bbC[V]_m$ is given by the monomials
$$
x_1^{a_1} \cdots x_n^{a_n}, \qquad a_i \geqslant 0,\quad
\sum_{i=1}^n a_i = m.
$$
For each such monomial, we have
$$
(g \cdot x_i)(\bsv) = x_i(\bsT_{g^{-1}} \bsv) = x_i\Pa{\sum_j c_j
\lambda_j^{-1} \bsv_j} = c_i \lambda_i^{-1} = \lambda_i^{-1}
x_i(\bsv),
$$
so $g$ acts on $x_i$ by multiplication by
$\lambda_i^{-1}$. Therefore,
$$
g \cdot (x_1^{a_1} \cdots x_n^{a_n}) = (\lambda_1^{-a_1} \cdots
\lambda_n^{-a_n}) x_1^{a_1} \cdots x_n^{a_n}.
$$
That is, the eigenvalues of $\cT_g^m$ are precisely all numbers of
the form $\lambda_1^{-a_1} \cdots \lambda_n^{-a_n}$, where $a_i
\geqslant 0$ and $\sum_{i=1}^n a_i = m$. Consequently,
$$
\Tr{\cT_g^m} = \sum_{\substack{a_1,\dots,a_n \geqslant 0 \\
\sum a_i= m}} \lambda_1^{-a_1} \cdots \lambda_n^{-a_n}.
$$
The Molien series is
$$
M_\G(q) := \sum_{m=0}^{\infty}
\dim\Pa{\bbC[V]_m^\G}q^m=\sum_{m=0}^{\infty}\Tr{\widetilde{\cT}_m}q^m.
$$
Using the preceding expression and interchanging summation and
integration (justified by uniform convergence for $\abs{q}<1$), we
obtain
\begin{eqnarray*}
M_\G(q) &=& \sum_{m=0}^{\infty} \Pa{\int_\G
\sum_{\substack{a_1,\dots,a_n \geqslant 0 \\ \sum a_i = m}}
\lambda_1^{-a_1} \cdots \lambda_n^{-a_n} \, \dif\mu(g)} q^m \\
&=& \int_\G \Pa{\sum_{m=0}^{\infty} \sum_{\substack{a_1,\dots,a_n
\geqslant 0
\\ \sum a_i = m}}\prod_{i=1}^n (q \lambda_i^{-1})^{a_i}} \dif\mu(g) \\
&=& \int_\G \Pa{\prod_{i=1}^n \sum_{a_i=0}^{\infty} (q \lambda_i^{-1})^{a_i}}\dif\mu(g) \\
&=& \int_\G \prod_{i=1}^n \frac{1}{1 - q \lambda_i^{-1}} \,
\dif\mu(g).
\end{eqnarray*}
Now observe that
$$
\prod_{i=1}^n (1 - q \lambda_i^{-1}) = \det\Pa{\I_V - q
\bsT_{g^{-1}}},
$$
because the eigenvalues of $\bsT_{g^{-1}}$
are precisely $\lambda_1^{-1}, \dots, \lambda_n^{-1}$. Hence
$$
\prod_{i=1}^n \frac{1}{1 - q \lambda_i^{-1}} = \det\Pa{\I_V - q
\bsT_{g^{-1}}}^{-1}.
$$
Since the Haar measure is invariant
under the inversion map $g \mapsto g^{-1}$, we have
$$
\int_\G \det\Pa{\I_V - q \bsT_{g^{-1}}}^{-1} \dif\mu(g) = \int_\G
\det\Pa{\I_V - q \bsT_g}^{-1} \dif\mu(g).
$$
Therefore, we conclude the Molien formula for a compact Lie group:
$$
M_\G(q)= \int_\G \frac{\dif\mu(g)}{\det\Pa{\I_V - q \bsT_g}}
$$
where
the integral is taken with respect to the normalized Haar measure on
$\G$. This gives the generating function for the dimensions of the
graded spaces of invariant polynomials.
\end{proof}

\section{Complete enumeration diagrams for cases (24), (26) and (31)}\label{graph}
The graphical enumeration for cases (24), (26), and (31) is	presented below. The diagrams are read from left to right, with braces indicating alternative continuations of a partial connection pattern. Solid and dashed lines denote $A$-type and $B$-type contractions, respectively. Configurations related by exchanging equivalent legs are displayed only once. Red diagrams denote candidates retained for the subsequent algebraic reductions, not necessarily members of the final generating set.

The labels indicate why a branch need not be expanded further. $R(0)$ denotes a vanishing contraction arising from identical vector attachments or a symmetric chain joining two legs of the same $\varepsilon$-node. $R(\ref{4c})$ denotes reduction of a chain containing at least five $\bsC$-nodes by Proposition~\ref{4c}. $R(\mathrm{det})$ denotes the determinant reduction when all three legs of an $\varepsilon$-node meet $\bsC$-nodes. $R(Q_1)$ and $R(Q_2)$ refer to the corresponding generic reductions in the proof of Theorem~\ref{fundamental}.

Some diagrams have unconnected legs. The reduction labels beside them indicate how the resulting contractions are reduced once these legs are connected. In particular, the $Q_1$ and $Q_2$ reductions apply when the remaining open ends are completed to vector paths, after the preceding chain and determinant reductions. The diagrams illustrate the case-by-case argument in the main text.

\begin{figure}[H] \centering
\includegraphics[ width=1.1\textwidth]{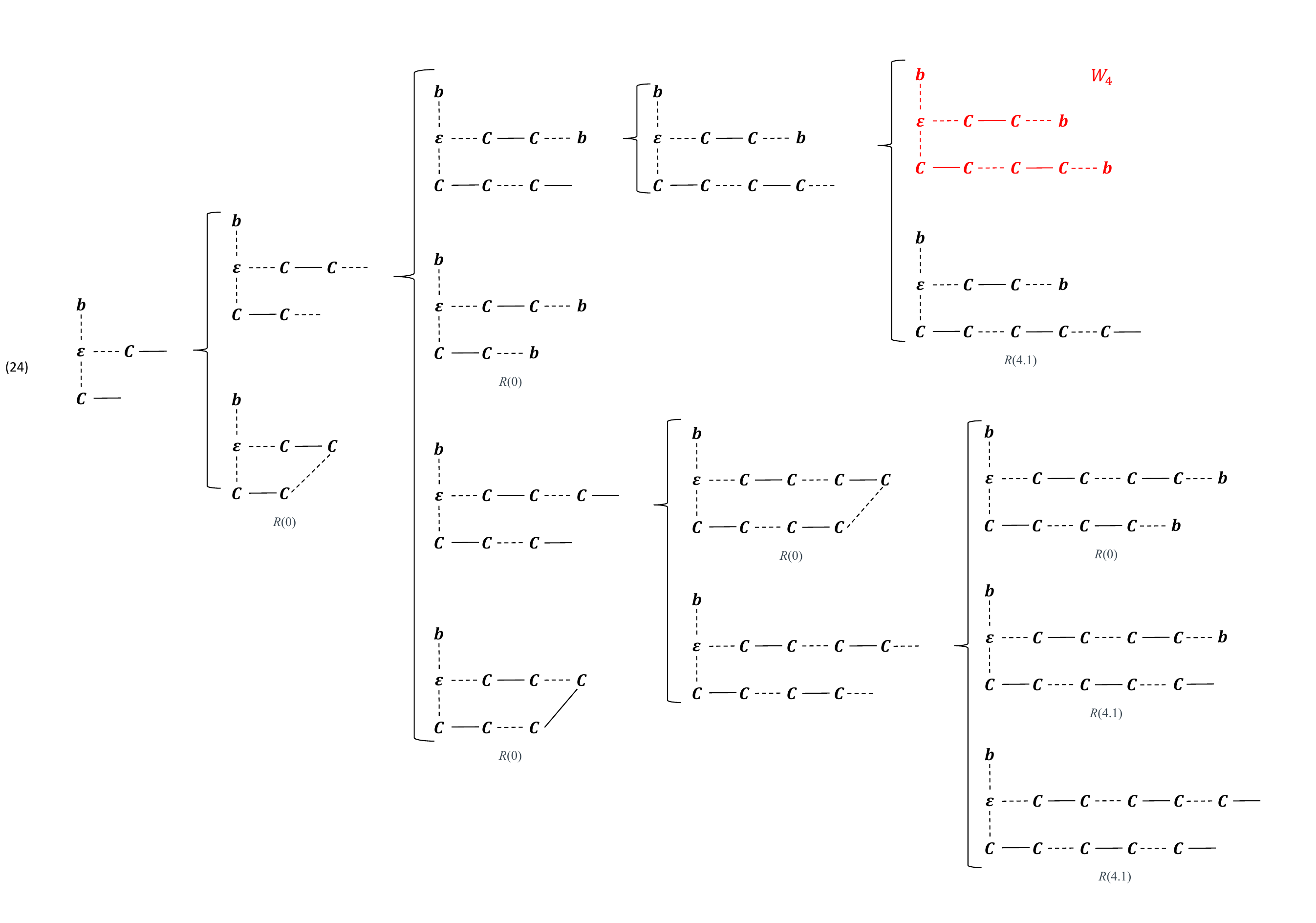}
\caption{The single candidate invariant in case (24).} \label{24}
\end{figure}

\begin{figure}[H] \centering
\includegraphics[trim=5 14 7 10, clip, width=1.1\textwidth]{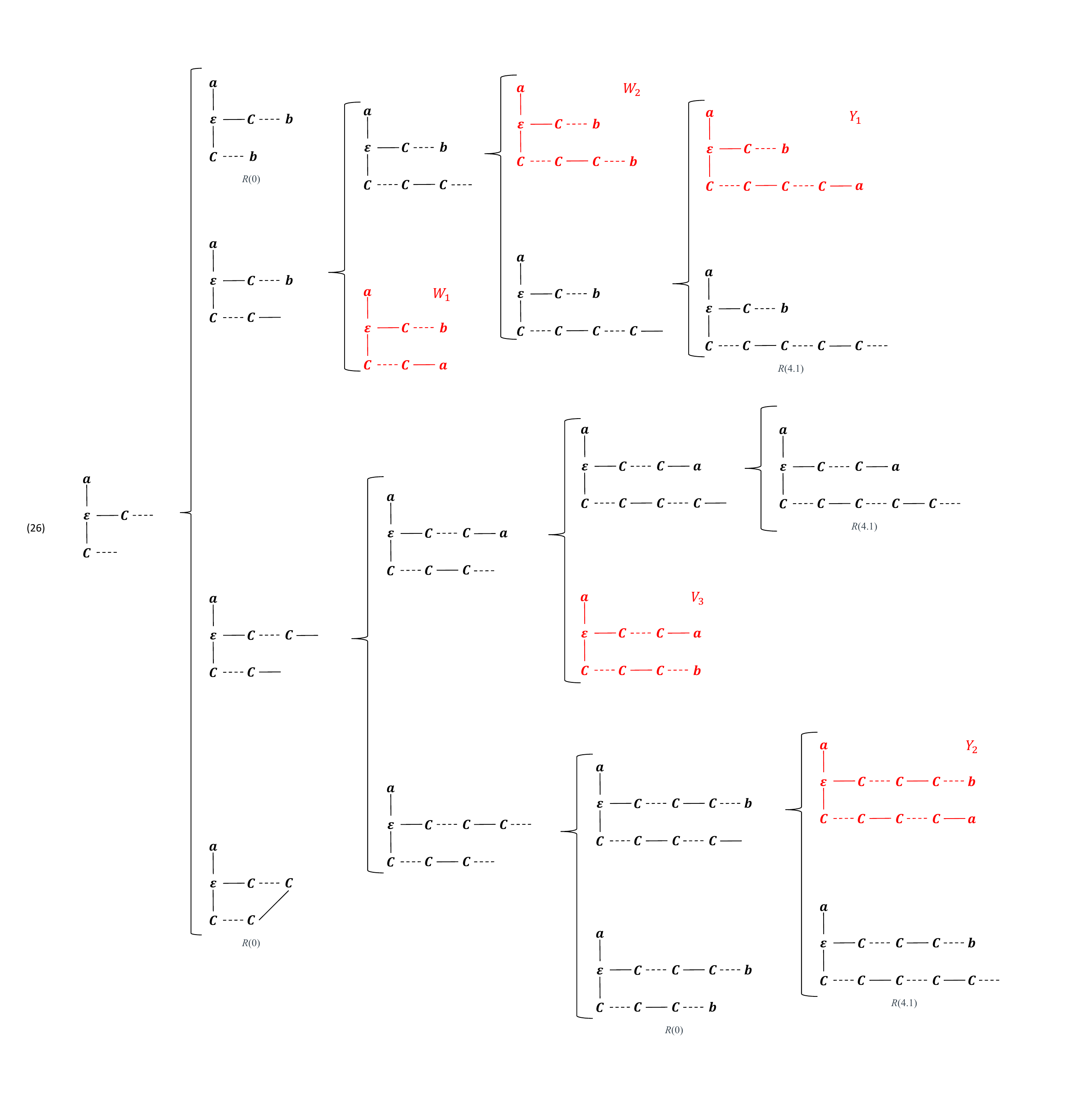}
\caption{In case (26), there are five candidate invariants.}
\label{26}
\end{figure}

\begin{figure}[H] \centering
\includegraphics[trim=5 14 7 33, clip, width=1\textwidth]{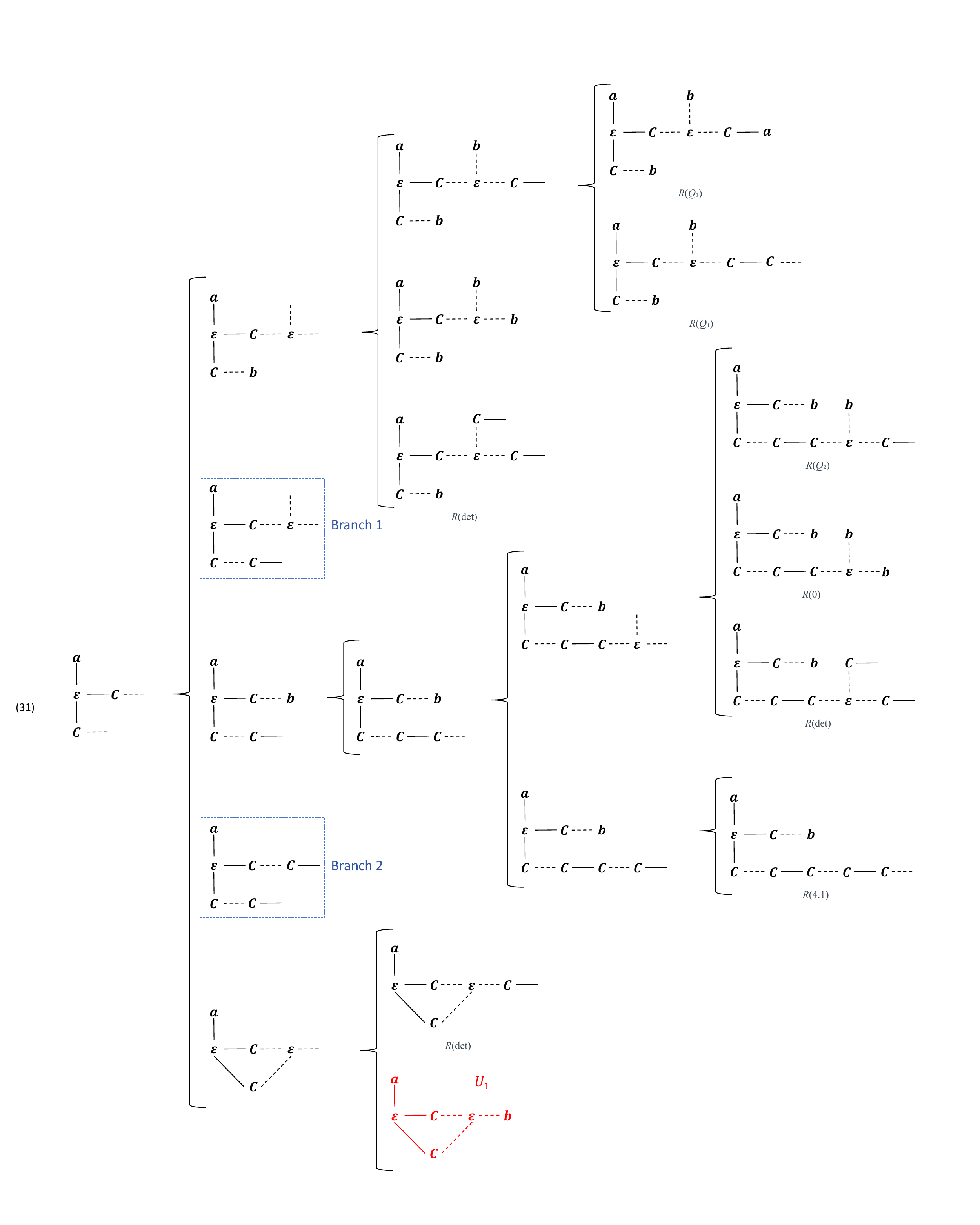}
\caption{ Graphical enumeration for case (31).
The branches enclosed by blue dotted rectangles are
expanded in Figures~\ref{31b} and \ref{31c}.} \label{31a}
\end{figure}

\begin{figure}[H] \centering
\includegraphics[trim=25 13 7 5, clip, width=1.2\textwidth]{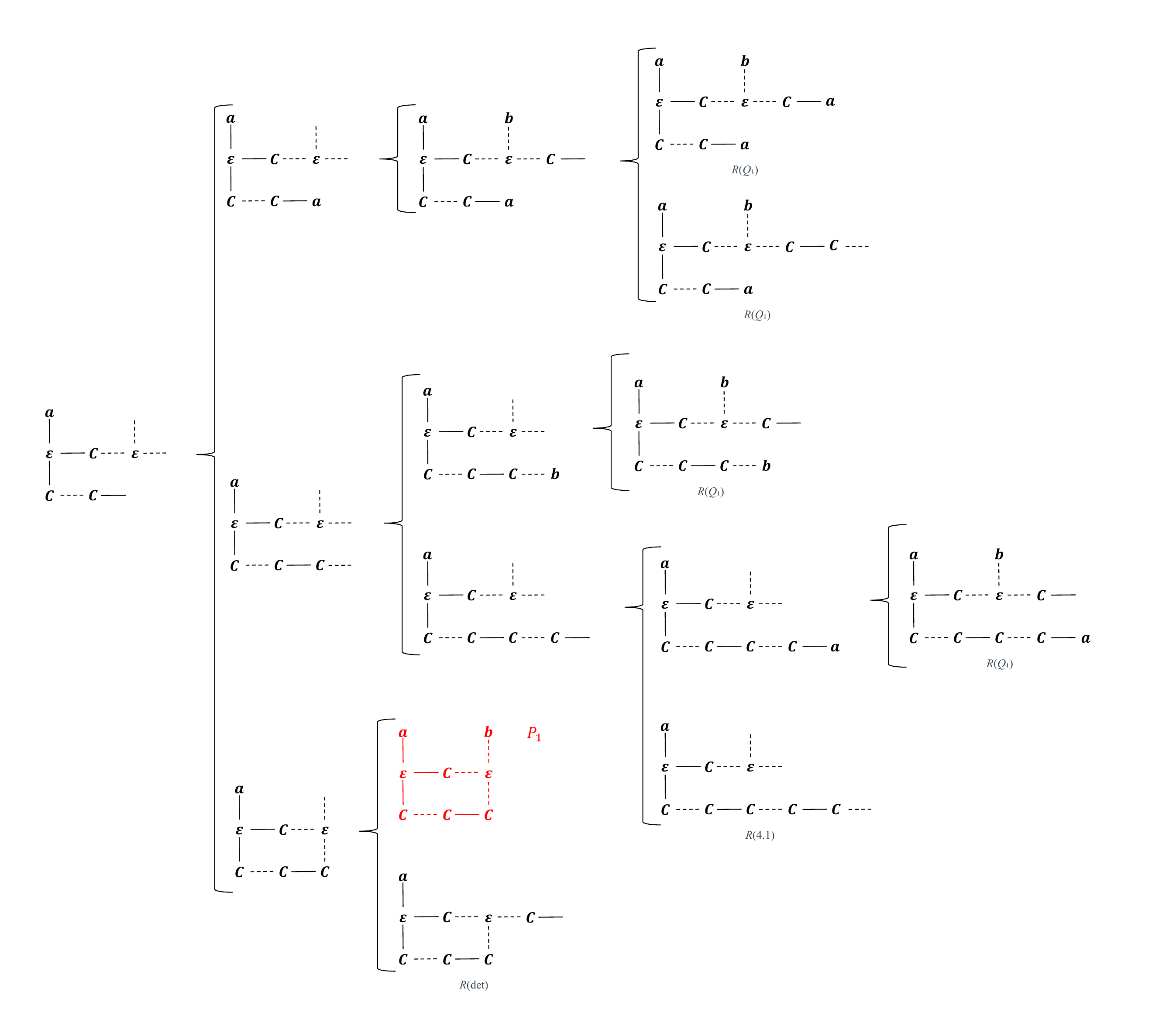}
\caption{ Graphical enumeration of branch 1 in Figure \ref{31a}.} \label{31b}
\end{figure}

\begin{figure}[H] \centering
\includegraphics[trim=44 22 37 55, clip, width=1.2\textwidth]{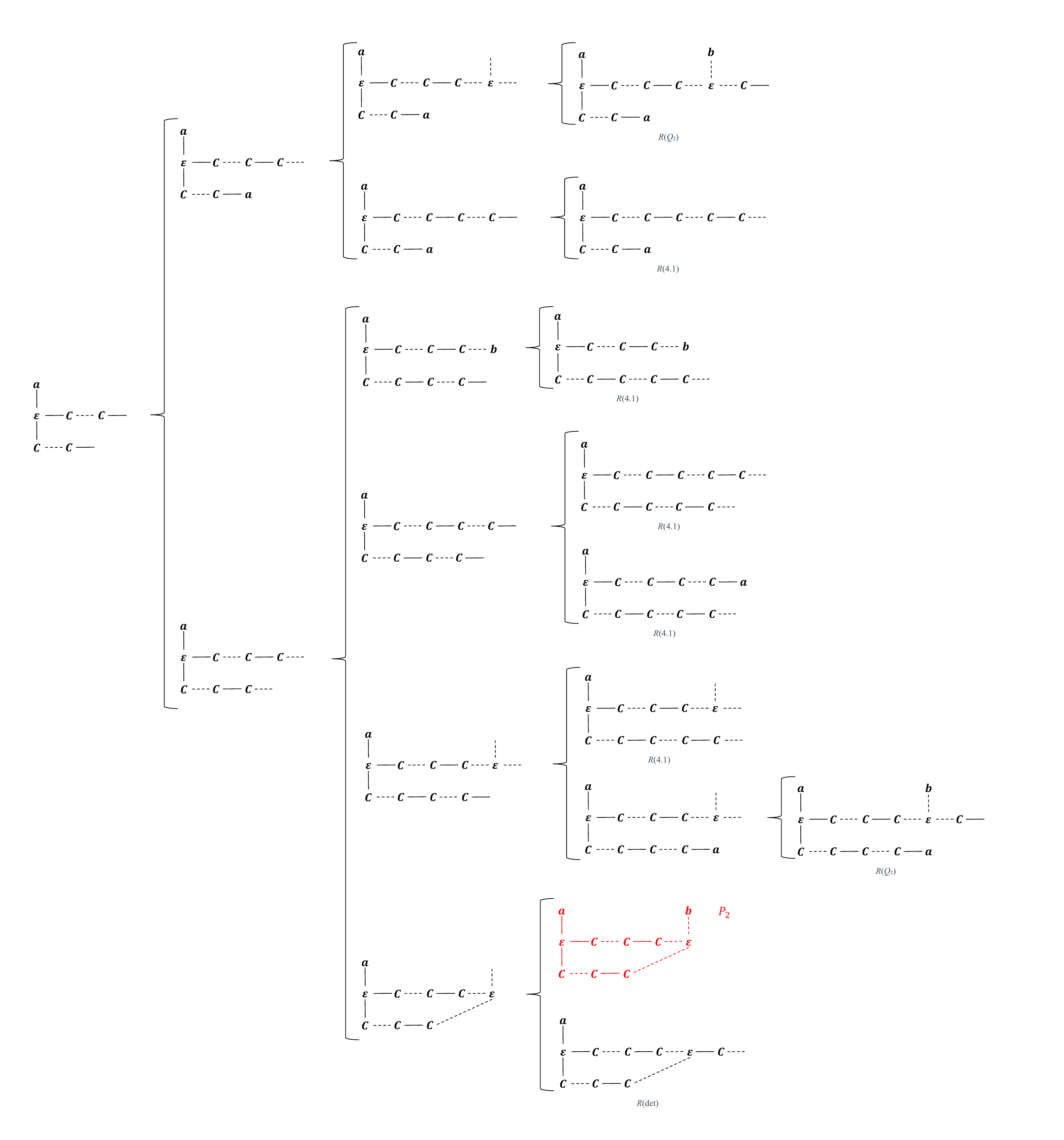}
\caption{ Graphical enumeration of branch 2 in Figure \ref{31a}.} \label{31c}
\end{figure}

\end{document}

%% file: qudit_preamble.tex
\newcommand{\jmp}{J. Math. Phys.~}

\newcommand{\jpa}{J. Phys. A~}

\newcommand{\pra}{Phys. Rev. A~}

\usepackage{young}
\usepackage{appendix}
\usepackage{epstopdf}
\usepackage{xcolor}
\usepackage{subfigure}
\usepackage[T1]{fontenc}
\usepackage[sc]{mathpazo}
\usepackage{amsmath}
\usepackage{amssymb}
\usepackage{graphicx,color}
\usepackage{framed}
\usepackage{multirow}
\usepackage{enumerate}
\usepackage{amsthm}
\usepackage{amsfonts,mathrsfs}
\usepackage{geometry} 
\usepackage{float}
\usepackage{eepic}
\usepackage{booktabs}
\usepackage{dashrule}
\usepackage[normalem]{ulem}
\usepackage{ifthen}
\usepackage[vcentermath]{youngtab}
\usepackage[unicode=true,pdfusetitle, bookmarks=true,bookmarksnumbered=false,bookmarksopen=false, breaklinks=false,pdfborder={0 0 0},backref=false,colorlinks=false] {hyperref}
\hypersetup{
colorlinks,linkcolor=myurlcolor,citecolor=myurlcolor,urlcolor=myurlcolor}
\definecolor{myurlcolor}{rgb}{0,0,0.7}

\usepackage{orcidlink}

\usepackage{color}

\usepackage{hyperref}
\hypersetup{pdfpagemode=UseNone}

\newcommand{\tinyspace}{\mspace{1mu}}

\newcommand{\op}[1]{\operatorname{#1}}

\newcommand{\abs}[1]{\left\lvert\tinyspace #1 \tinyspace\right\rvert}

\renewcommand{\det}{\operatorname{det}}
\renewcommand{\t}{{\scriptscriptstyle\mathsf{T}}}

\renewcommand{\vec}{\op{vec}}
\newcommand{\im}{\op{im}}

\def\GL{\mathsf{GL}}
\def\SL{\mathsf{SL}}
\def\SO{\mathsf{SO}}
\def\SU{\mathsf{SU}}

\def\Res{\mathrm{Res}}

\def \dif {\mathrm{d}}
\def \diag {\mathrm{diag}}

\def \im {\mathrm{Im}}

\def\I{\mathbb{1}}

\newenvironment{mylist}[1]{\begin{list}{}{
    \setlength{\leftmargin}{#1}
    \setlength{\rightmargin}{0mm}
    \setlength{\labelsep}{2mm}
    \setlength{\labelwidth}{8mm}
    \setlength{\itemsep}{0mm}}}
    {\end{list}}

\def\ot{\otimes}

\newcommand{\inner}[2]{\langle #1 , #2\rangle}

\newcommand{\out}[2]{| #1\rangle\langle #2 |}
\newcommand{\Inner}[2]{\left\langle #1 , #2\right\rangle}
\newcommand{\innerm}[3]{\left\langle #1 \left| #2 \right| #3 \right\rangle}
\newcommand{\defeq}{\stackrel{\smash{\textnormal{\tiny def}}}{=}}

\newcommand{\End}{\mathrm{End}}

\newcommand{\pauli}{\boldsymbol{\sigma}}

\newcommand{\pa}[1]{(#1)}
\newcommand{\Pa}[1]{\left(#1\right)}

\newcommand{\Br}[1]{\left[#1\right]}
\newcommand{\set}[1]{\{#1\}}
\newcommand{\Set}[1]{\left\{#1\right\}}

\newcommand{\ket}[1]{|#1\rangle}

\DeclareMathOperator{\trace}{Tr}
\newcommand{\ptr}[2]{\trace_{#1}\pa{#2}}
\newcommand{\Ptr}[2]{\trace_{#1}\Pa{#2}}
\newcommand{\tr}[1]{\ptr{}{#1}}
\newcommand{\Tr}[1]{\Ptr{}{#1}}

\newcommand{\Abs}[1]{\left|\tinyspace#1\tinyspace\right|}

\def\cF{\mathcal{F}}\def\cH{\mathcal{H}}

\def\cR{\mathcal{R}}\def\cT{\mathcal{T}}
\def\cX{\mathcal{X}}

\def\bsA{\boldsymbol{A}}\def\bsB{\boldsymbol{B}}\def\bsC{\boldsymbol{C}}
\def\bsF{\boldsymbol{F}}
\def\bsM{\boldsymbol{M}}\def\bsN{\boldsymbol{N}}\def\bsO{\boldsymbol{O}}
\def\bsS{\boldsymbol{S}}\def\bsT{\boldsymbol{T}}
\def\bsX{\boldsymbol{X}}\def\bsY{\boldsymbol{Y}}

\def\bsa{\boldsymbol{a}}\def\bsb{\boldsymbol{b}}\def\bse{\boldsymbol{e}}

\def\bsu{\boldsymbol{u}}\def\bsv{\boldsymbol{v}}\def\bsw{\boldsymbol{w}}\def\bsx{\boldsymbol{x}}\def\bsy{\boldsymbol{y}}
\def\bsz{\boldsymbol{z}}

\def\G{\textsf{G}}                  

\def\T{\textsf{T}}
\def\U{\textsf{U}}

\def\bbC{\mathbb{C}}

\def\bbR{\mathbb{R}}

\def\sfG{\mathsf{G}}

\def\sfU{\mathsf{U}}

\newtheorem{thrm}{Theorem}[section]
\newtheorem{lem}[thrm]{Lemma}
\newtheorem{prop}[thrm]{Proposition}
\newtheorem{cor}[thrm]{Corollary}
\theoremstyle{definition}
\newtheorem{definition}[thrm]{Definition}

\numberwithin{equation}{section}

\newcounter{questionnumber}